\documentclass[11pt, a4paper]{kybernetika}
\usepackage{amsfonts, amsthm, array, bm, colortbl, comment, enumitem, fancyhdr, geometry, intcalc, lscape, makecell, mathpazo, mathtools, mdframed, microtype, multicol, multirow, pgfplots, scrextend, stmaryrd, tabularx, tikz, xcolor, xfrac}
\usepackage[british]{babel}
\usepackage[backend=biber, maxnames=9, minnames=1, sorting=nty, style=numeric]{biblatex}
\usepackage[Sonny]{fncychap}
\usepackage[T1]{fontenc}
\usepackage[utf8]{inputenc}
\usepackage[semibold]{sourcesanspro}
\usepackage[absolute]{textpos}
\usepackage[normalem]{ulem}

\usetikzlibrary{arrows, arrows.meta, decorations.markings, positioning, shapes.geometric, trees}

\newtheorem{theorem}{Theorem}[section]
\newtheorem{lemma}[theorem]{Lemma}
\newtheorem{proposition}[theorem]{Proposition}

\newtheorem{remark}[theorem]{Remark}
\newtheorem{fact}[theorem]{Fact}
\newtheorem{example}[theorem]{Example}
\newtheorem{definition}[theorem]{Definition}

\DeclareUnicodeCharacter{0156}{\"U}
\DeclareUnicodeCharacter{0188}{\"u}
\DeclareUnicodeCharacter{0240}{\.I}
\DeclareUnicodeCharacter{0241}{\i}
\DeclareUnicodeCharacter{0301}{\"o}
\DeclareUnicodeCharacter{0308}{\"i}
\DeclareUnicodeCharacter{1EF3}{\`y}

\tikzset{
	treenode/.style = {align=center, text centered,
		font=\sffamily},
	siybey/.style = {treenode, circle, font=\sffamily\bfseries, draw=black},
}

\newcommand{\prob}{\mathrm{Prob}}

\fancyheadoffset{0.005\textwidth}
\usepackage{hyperref}

\begin{document}
	\pagestyle{myheadings}
	\title{Bounds on Guessing Numbers and Secret Sharing \\ Combining Information Theory Methods}
	\author{Emirhan G\"urpınar}
	\contact{Em\.irhan}{G\"urpınar}{}{emirhangurpinar@mailfence.com}
	\markboth{E. GURPINAR} {Bounds on Guessing Numbers and Secret Sharing}
	\maketitle
	\begin{abstract}
		This paper is on developing some computer-assisted proof methods involving non-classical inequalities for Shannon entropy.
		
		Two areas of the applications of information inequalities are studied: Secret sharing schemes and hat guessing games. In the former a random secret value is transformed into shares distributed among several participants in such a way that only the qualified groups of participants can recover the secret value. In the latter each participant is assigned a hat colour and they try to guess theirs while seeing only some of the others'. The aim is to maximize the probability that every player guesses correctly, the optimal probability depends on the underlying sight graph. We use for both problems the method of non-Shannon-type information inequalities going back to Z. Zhang and R. W. Yeung. We employ the linear programming technique that allows to apply new information inequalities indirectly, without even writing them down explicitly. To reduce the complexity of the problems of linear programming involved in the bounds we extensively use symmetry considerations. Using these tools, we improve lower bounds on the ratio of key size to secret size for the former problem and an upper bound for one of the ten vertex graphs related to an open question by Riis for the latter problem.
	\end{abstract}	
	
	
	
	\section{INTRODUCTION}
	
	The aim of this paper is to show how the techniques of computer-assisted proofs for information inequalities (for the Shannon entropy) can be used in various applications.
	Each ingredient of our approach has already been known, but we argue that a properly chosen combination of these methods is quite powerful, to the point that we can improve several previously known bounds.
	We apply the techniques to two targets: we prove lower bounds for the efficiency of \emph{secret sharing schemes} (for several specific access structures) and lower bounds for \emph{hat guessing games}.
	Our new bounds are proven using heavy computations and it seems that the same results would be very hard to achieve manually, without a computer. Our main goal is to show the efficiency of the right combination of technical tools; this is why we have deliberately chosen problems (in secret sharing and in hat-guessing games) that were already studied earlier by other researchers, so that we can compare our results with previously known bounds.
	
	We go on with a brief review of the fields of secret sharing and hat guessing games in which we apply our techniques.
	
	The notion of secret sharing introduced by Shamir \cite{shamir1979share} and Blakley \cite{blakley1979safeguarding}, is nowadays pretty standard in cryptography. In what follows the motivation and the basic definition of secret sharing is recalled. Imagine that we want to share a secret between some participants in such a way that	
	\begin{itemize}
		\item some coalitions (subsets of participants), the \emph{authorized} ones, can reconstruct the secret combining their shares;
		\item still the other coalitions, the \emph{forbidden} ones, get no information about the secret.
	\end{itemize}
	One can easily imagine practical situations when this ability is useful, and Shamir's famous secret sharing scheme deals with the case when all sufficiently large groups (at least $t$ participants for some threshold $t$) are authorized while small groups with less than $t$ participants are forbidden. Given the description of the authorized and forbidden coalitions for a set of participants, we want to find how small the maximal share size can be made, with respect to the size of the secret. In general this is an open problem, thus we look for lower bounds on this quantity as frequently done in the literature. Such questions are of interest in their own right. Also, they can be used as benchmarks for the techniques based on information inequalities.
	
	In this article we study several particular access structures. We improve some previously known lower bounds for their information ratios. The previously known bounds were obtained using the Ahlswede--K\"orner (AK) lemma in \cite{farras2020improving}. We use a different technique -- the general version of the \emph{copy lemma} \cite{zhang1998characterization} combined with symmetry considerations. For every given access structure, we reduce the question of the information ratio to a linear problem (of very large size) and then use linear programming solvers to obtain a lower bound for this information ratio. Our results on secret sharing are summarized in Theorem~\ref{thm:spanish-bounds} and compared with the previously known bounds in Table~\ref{table-sss-results}.
	
	The \emph{hat guessing games} is a family of recreational mathematics problems \cite{winkler2002games}, some variants of which are known to be connected to coding theory \cite{robinson2001mathematicians}. Each player gets a hat of some colour (invisible to them) and has to guess this colour (knowing only the colours of hats they see). There are many versions of these games \cite{butler2009hat,ma2011new}. We consider a version introduced by Riis in \cite{riis2006utilising,riis2007information} as it is connected with some problems in network information theory. In this version the visibility (who can see whose hats) is determined by a graph named the sight graph. The challenge is to maximize the probability that each player guesses their hat colour correctly. No communication between players is allowed during the game, but a strategy can be agreed upon before the game. 
	
	This problem for an arbitrary graph remains open. Moreover there exists a specific graph with 10 vertices for which the question is open (the single smallest such undirected graph). We improve the upper bound on the probability of `correct guessing' for this graph. Our main result here is given in Theorem~\ref{th:R-} and compared with the previously known bound at the beginning of Section~\ref{Ch:HGG}.
	
	To bound the quantities that arise in both of these problems (secret sharing and hat guessing games), we use a combination of several techniques. To prove the bounds, we use non-classical inequalities for entropy (non-Shannon-type inequalities). We derive the necessary inequalities with the \emph{copy lemma}. More precisely, we use these new inequalities indirectly, without writing them explicitly. To this end, we combine the copy lemma with the linear programming approach. To decrease the complexity of the linear program and improve the results we use symmetry considerations (the symmetries of the authorized and forbidden coalitions for problems of secret sharing and the symmetries of the sight graphs for the hat guessing games). Each of these techniques has already been known. However their combination proves to be so efficient that we improve some known bounds for these problems. Our improved bounds are given in the Results sections (in Section~\ref{SSS:R} and Section~\ref{s:our-res}).
	
	In the next section we give preliminaries and explain the context in more detail. In Section~\ref{symmetry-section} we explain the symmetries and compute the symmetry groups of our problems. In the following sections we formulate the particular problems on which we apply our method.
	
	\section{Preliminaries}\label{ch:preliminaries}
	
	\subsection{Entropy Region}
	
	\begin{definition}[Entropy Vector]
		Let $X=(X_i)_{i\in\llbracket1,n\rrbracket}$ be a sequence of jointly distributed random variables with a finite range. We denote by $h_X$ the vector, the coordinates of which are the values of Shannon entropy for all sub-tuples of $X$. This vector is called the \emph{entropy vector} (also known as \emph{the entropy profile}) of $X$. Note that it consists of $2^n-1$ real components $h_I=H((X_i)_{i\in I})$ for each ${\varnothing\neq I\subseteq\llbracket1,n\rrbracket}$, so it is in $\mathbb{R}^{2^n-1}$.
	\end{definition}
	
	\begin{definition}[Entropy Region]
		For $n>0$, the set of all entropy vectors of dimension $2^n-1$ (for every distribution of $n$-tuples of random variables) is called the \emph{entropy region}. Following \cite{zhang1998characterization}, we use the notation $\Gamma_n^*\subset\mathbb{R}^{2^n-1}$ for it.
	\end{definition}
	
	\begin{definition}[Almost Entropic]
		The closure of $\Gamma_n^*$ is noted $\overline{\Gamma_n^*}$, and its elements are called \emph{almost entropic} vectors. Any non-strict inequality that is satisfied by all elements of $\Gamma_n^*$ is also satisfied by all elements of $\overline{\Gamma_n^*}$ by limit.
	\end{definition}
	
	\begin{remark}\label{rem:convex-cone}
		$\overline{\Gamma_n^*}$ is a convex cone \cite{zhang1997non}. In particular it is invariant under multiplication by a non-negative scalar in $\mathbb{R}_+$.
	\end{remark}
	
	So $\overline{\Gamma_n^*}$ is defined solely by the linear inequalities satisfied by $\Gamma_n^*$.
	
	The characterization of $\Gamma_n^*$ and that of its closure are open problems.
	
	\subsection{Information Inequalities For Entropy}\label{s:inf-ineq}
	
	\begin{definition}[Information Inequality]
		An \emph{information inequality} for entropy is a linear inequality for the entropy quantities of jointly distributed random variables with real coefficients.
	\end{definition}
	By definition of $\Gamma_n^*$, the \emph{information inequalities} for $n$ random variables are exactly the linear inequalities for $2^n-1$ coordinates that are true for all vectors in $\Gamma_n^*$.
	
	The first universally true information inequalities were given in the seminal paper \cite{shannon1948mathematical} by Shannon.
	
	\begin{definition}[Shannon and Shannon-type inequalities]\label{def:Shannon-ineq}
		Let us denote $(X_i)_{i\in I}$ by $X_I$ in short. The inequalities of the form $$I(X_I:X_J|X_K)\ge0$$ are called \emph{Shannon inequalities}. They can be expanded as
		$$H(X_{I\cup K}) + H(X_{J\cup K}) \ge H(X_{I\cup J\cup K}) + H(X_{K}).$$
		The inequalities that are linear combinations with positive coefficients of Shannon inequalities are called \emph{Shannon-type (classical) inequalities}. 
	\end{definition}
	
	The set of vectors with $2^n-1$ coordinates (not necessarily entropic) satisfying all classical inequalities is noted $\Gamma_n$.
	
	Note that $\Gamma_n^*\subset\overline{\Gamma_n^*}\subset\Gamma_n$, that $\Gamma_n^*$ is closed under addition and that $\overline{\Gamma_n^*}$ is a convex cone \cite{zhang1997non}.
	
	\begin{definition}[Elemental Information Inequalities]
		Let $X_1,\ldots X_n$ be random variables. The inequalities of the form $$I(X_i:X_j|X_K)\ge0$$ where $i\neq j$ and $K\subseteq\llbracket1,n\rrbracket\setminus\{i,j\}$ or $i=j$ and $K=\llbracket1,n\rrbracket\setminus\{i\}$ are called \emph{elemental information inequalities} or shortly \emph{elemental inequalities}.
	\end{definition}
	\begin{fact}\label{prop:elem}
		Elemental inequalities for $n$ variables imply all Shannon inequalities for $n$ variables by linear combinations with positive coefficients, see \cite[Chapter~13]{yeung2002first}.
	\end{fact}
	\begin{remark}\label{rem:number-of-inequalities}
		There are $\Theta(n^2\cdot 2^n)$ elemental inequalities (precisely $\binom n22^{n-2}+n$) for $n$ variables,
		whereas there are $\Theta(5^n)$ different Shannon inequalities.
	\end{remark}
	
	\subsection{How To Prove An Information Inequality}
	
	Consider a Shannon-type inequality, i.e. a linear combination of Shannon inequalities. Here the situation is simple: for a given number $m$ of random variables (or strings) we have $2^m-1$ entropic quantities and can write down all the Shannon inequalities for these quantities. Then we want to know whether the inequality in question is a non-negative linear combination of Shannon inequalities. This is a classical question of linear programming, a linear program solver finds out whether it is true or not. One should have in mind, however, that the dimension of the linear program grows exponentially in $m$, so the linear program could be quite large. Only very simple cases could be treated by hand, and quite soon we bump into a system of linear inequalities that is too large even for computer tools. To make it smaller, we can use dependencies between Shannon inequalities for different tuples and omit some inequalities that can be derived from the elemental ones.
	
	Let $I_i(X_1,\ldots X_m)$ be linear combinations of entropies of subsets of $\{X_1,\ldots X_m\}$ for $i\in\mathcal{I}$. Statements of the form
	$$\wedge_{i\in\mathcal{I}}I_i(X_1,\ldots X_m)\ge0\implies I(X_1,\ldots X_m)\ge0.$$
	are called conditional or constraint inequalities. They often appear in applications. The same technique of linear programming can be applied by adding the conditions 
	$$I_i(X_1,\ldots X_m)\ge0,i\in\mathcal{I}$$
	to the linear program.
	
	In this way we can derive the Shannon-type inequalities starting from Shannon inequalities. Of course, we may as well add some known non-Shannon-type inequalities to the list of inequalities that we combine.
	
	To get and prove non-Shannon-type inequalities, the most common tool is the \emph{copy lemma} that we are going to formulate now. Let us split all the random variables $X_1,\ldots X_m$ into two groups $A_1,\ldots A_k$ and $B_1,\ldots B_\ell$ (in an arbitrary way). We can assume that $A_1,\ldots A_k$ are sampled first according to their marginal distribution and then $B_1,\ldots B_\ell$ are sampled according to their conditional distribution. In this way we get the same distribution $A_1,\ldots A_k, B_1,\ldots B_\ell$, so nothing new is obtained yet. But we can, for the same values of $A_1,\ldots A_k$, consider another independent sample $B_1',\ldots B_\ell'$ using the same conditional distribution. Then, instead of $k+\ell$ variables $A_1,\ldots A_k,B_1,\ldots B_\ell$ that we started with, we get a joint distribution for $k+2\ell$ variables
	\[
	A_1,\ldots A_k,\ B_1,\ldots B_\ell,\ B_1',\ldots B_\ell'
	\]
	that have the following properties:
	\begin{itemize}
		\item 
		the distribution of $A_1,\ldots A_k,B_1,\ldots B_\ell$ is the same as before;
		\item
		the distribution of $A_1,\ldots A_k,B_1',\ldots B_\ell'$ is the same as for $A_1,\ldots A_k,B_1,\ldots B_\ell$;
		\item
		the tuples $B_1,\ldots B_\ell$ and $B_1',\ldots B_\ell'$ are independent given $A_1,\ldots A_k$.
	\end{itemize}
	
	Formally we state the copy lemma in the following two forms. $X$ below corresponds to $A_1,\ldots A_k$ above. $Y$ and $Z$ correspond to a partition of $B_1,\ldots,B_\ell$ into the variables we sample only at first step and the variables we resample at the second step respectively. This distinction is useful for not increasing the number of total variables of the linear program we use.
	
	\begin{lemma}[Copy Lemma \cite{zhang1998characterization,dougherty2011non}]
		Let $X,Y,Z$ be three jointly distributed random vectors.
		\begin{enumerate}
			\item There exists a random vector $Z'$ such that
			\begin{itemize}
				\item $X,Z$ and $X,Z'$ are identically distributed,
				\item $Z'$ and $Y,Z$ are independent given $X$.
			\end{itemize}
			\item There exists a random vector $Z'$ such that
			\begin{itemize}
				\item every sub-vector of $X,Z$ has the same entropy as the sub-vector of $X,Z'$ that consist of the same coordinates,
				\item $H(Z':Y,Z|X)=0$.
			\end{itemize}
		\end{enumerate}
	\end{lemma}
	
	Item~1 above is the probabilistic statement, which implies the entropic statement item~2. We use the latter in our applications.
	
	Now we can use linear programming to derive consequences of all the Shannon inequalities for all variables ($X,Y,Z,Z'$) and the equalities that are guaranteed by our constructions. Zhang and Yeung \cite{zhang1998characterization} discovered that this way \emph{we get new inequalities that include only original variables $X,Y,Z$}. By `new', we mean inequalities that are non-Shannon-type, i.e. they are not linear combinations of Shannon inequalities for original variables. Then these new inequalities can be used explicitly, by adding them to the list of Shannon inequalities, or implicitly.
	
	We also extensively use the symmetries of the problem, which guarantee that an optimal solution can be found among the symmetric ones. This helps to reduce the dimension of the linear program and let the solver work faster.
	
	We discuss these tricks in detail in the corresponding sections.
	
	\subsection{Preliminaries of Secret Sharing}
	
	Secret sharing was independently introduced in \cite{blakley1979safeguarding} and \cite{shamir1979share}. These original papers studied a class of secret sharing schemes which are now called \emph{threshold schemes}. A more general definition of secret sharing was introduced by Ito, Saito and Nishizeki \cite{ito1989secret}.
	One of the relatively recent surveys on the topic is \cite{beimel2011secret}, see also the lecture notes \cite{padro2012lecture}. 
	
	The aim of secret sharing is to distribute a secret between participants by giving each of them a personal share, such that every `accepted' coalition of participants combining their shares can reconstruct the secret, whereas no `forbidden' coalition of participants can get any information about it. One of the main problems in the field is to compute, for a given access structure, the optimal information ratio, which measures how large must be the shares for a secret of given size. In general, this problem remains widely open.
	The researchers working in this field keep improving upper and lower bounds for the information ratio of many non-trivial access structures, however, there still is an exponential gap between lower and upper bounds. 
	
	Let us give the standard formal definitions of access structure and secret sharing scheme.
	\begin{definition}[Access Structure]
		An \emph{access structure} for secret sharing among $n$ participants $\llbracket1,n\rrbracket$ partitions all coalitions  into two disjoint classes $\mathcal{A}$ (accepted coalitions) and $\mathcal{B}$ (forbidden coalitions) in such a way that every super-set of an accepted coalition is accepted.
	\end{definition}
	\begin{remark}
		An access structure is determined by the family of its minimal (by inclusion) accepted coalitions (denoted $\min\mathcal{A}$). A coalition is accepted if it contains a minimal accepted subset and forbidden otherwise.
	\end{remark}
	
	For a given access structure, we define the notion of a secret sharing scheme among $n$ participants requiring that every accepted coalition knows the secret and no forbidden one has any information about it.
	
	\begin{definition}[Secret Sharing Scheme]\label{def:SSS}
		We formally define a \emph{secret sharing scheme} for a given access structure with participants $1,\ldots n$ as a joint distribution (a tuple of random variables) $(S_0,S_1,\ldots S_n)$ satisfying the following conditions for each coalition $J$:
		\begin{equation}\label{eq:secret-sharing}
			\begin{array}{ll}
				H(S_0|S_J)=0,&\text{ if }J\text{ is an accepted coalition},\\
				H(S_0|S_J)=H(S_0),&\text{ if }J\text{ is a forbidden coalition}.
			\end{array}
		\end{equation}
		
		The random variable $S_0$ is called the \emph{secret key}, and $S_j$ for $j\in\llbracket1,n\rrbracket$ are the \emph{shares} given to each party and $S_J$ is short for $(S_j)_{j\in J}$.
	\end{definition}			
	
	Informally, this definition describes the process of generating simultaneously both the secret $S_0$ and shares $S_1,\ldots S_n$. We assume that the secret is non-trivial ($H(S_0)>0$).
	
	\smallskip
	
	Ito, Saito and Nishizeki proved in \cite{ito1989secret} that for every access structure there exists a secret sharing scheme.
	
	\begin{fact}[\cite{ito1989secret}]\label{prop-all-sss-acc-str}
		Every access structure admits a secret sharing scheme.
	\end{fact}
	
	%
	%
	
	Benaloh and Leichter \cite{benaloh1988generalized} noted that the construction of the proof is a special case of a more general one that starts from the monotone boolean function that describes the access structure. 
	
	\begin{definition}[Information Ratio]
		The \emph{information ratio of a secret sharing scheme} is the proportion of the size of the largest key to that of the secret, i.e. $\max_{i}\frac{H(S_i)}{H(S_0)}$.
		
		The \emph{information ratio of an access structure} is the infimum of the information ratios of the secret sharing schemes on the access structure.
	\end{definition}
	
	Note that in this definition we do not restrict the size of $S_0$; it may happen that schemes with good ratio exist only for large $S_0$. 
	
	Given an access structure, we may ask what its information ratio is. The ratio is at least $1$ for every non-trivial perfect access structure (if there is at least one accepted coalition). Indeed, if $S_1,\ldots S_k$ are shares for a minimal accepted coalition, and $S_0$ is the secret; then $$I(S_0:S_1,\ldots S_{k-1})=0$$ $$\text{and }I(S_0:S_1,\ldots S_k)=H(S_0).$$ The chain rule guarantees that $H(S_k)\ge H(S_0)$. This motivates the following definition.
	
	\begin{definition}[Ideal Secret Sharing]
		A secret sharing scheme with the information ratio $1$ is called an \emph{ideal} secret sharing scheme. An access structure that admits an ideal secret sharing scheme is called an \emph{ideal} access structure.
	\end{definition}
	
	The secret sharing scheme proposed by Shamir is ideal. It works as follows. 			
	Let $n$ be the number of participants and $d$ be the threshold. Consider a finite field $K$ with at least $n+1$ elements, and fix $n+1$ of them. Generate randomly a polynomial of degree at most $d-1$ with coefficients in $K$. Give the values of $P$ on $n$ chosen elements of $K$ to $n$ participants, and let the secret be the value of the polynomial on the last chosen element of $K$. Together every $d$ participants can reconstruct the polynomial and find the secret; whereas for any coalition of $d-1$ participants, all values of the secret have equal probabilities (there is exactly one polynomial for every choice of the secret in $K$). It is clear that smaller coalitions have no information about the value of the secret. The entropy of every share and that of the secret are $\log_2|K|$, so the information ratio of this scheme is $1$.
	
	On one hand some access structures, such as the threshold one, are ideal. On the other hand the general construction of a secret sharing scheme for arbitrary access structures gives only an exponential upper bound for information ratio. It may happen that for some access structures with $n$ participants the information ratio is indeed exponential in $n$. However,the currently known lower bounds are much weaker: 
	it was proven by Csirmaz \cite{csirmaz1997size} that there exist an $n$-participant access structure the information ratio of which is at least $n/\log_2n$.
	
	The common approach to prove a lower bound for the information ratio of a certain access structure is to use the technique of information inequalities. We write down the equalities~\eqref{eq:secret-sharing} and all Shannon-type inequalities for the involved random variables and then use linear programming to combine these equalities and inequalities to derive a result
	\begin{equation} \label{eq:info-ratio}
		\max_iH(S_i)\ge r\cdot H(S_0)
	\end{equation}
	for a certain real number $r$. If we succeed, this means that the information ratio of this access structure is at least $r$. Such an argument can be found in \cite{capocelli1993size} among others.
	
	This simple scheme can be improved. One can add non-Shannon-type inequalities as well; these additional constraints may help to prove \eqref{eq:info-ratio} with a larger value of $r$. Proofs following this scheme can be found, for instance in \cite{beimel2008matroids} and \cite{metcalf2011improved}. However, we do not follow this scheme and do not explicitly add non-Shannon-type inequalities to our linear program. Instead, we do as follows:
	
	\begin{enumerate}
		\item We write the conditions \eqref{eq:secret-sharing};
		\item instead of looking for non-Shannon-type inequalities for the variables that appear in these conditions, we apply one or several times the copy lemma, thus get some new random variables and some equalities for their entropies;
		\item then we write down \emph{only Shannon-type inequalities} but for \emph{all} the involved random variables and then deduce $\eqref{eq:info-ratio}$ for some specific $r$. 
	\end{enumerate}
	In this way we implicitly use the non-Shannon-type inequalities for old variables, inequalities that can be derived by using the copy lemma to get new variables.	
	This type of argument is discussed in \cite{gurpinar2019use}. A similar approach (with the AK lemma instead of the copy lemma) was used earlier in \cite{farras2018improving} and later in \cite{bamiloshin2021common}.
	
	\subsection{Preliminaries of Hat Guessing Games}
	
	\emph{Hat guessing game} has been a well known recreational mathematics problem with many variants. In this article we are interested in the variant introduced by Riis \cite{riis2006utilising,riis2007information}. It is connected to some problems in network coding, and it uses some concepts from graph theory.
	
	Let $G=(V,E)$ be a finite directed graph where $V$ is the set of vertices and $E\subseteq V\times V$ the set of directed edges. Note that all the graphs we consider in this article are loopless. Let $s>1$ be an integer, and let us denote $A_s$ the set of colours. The game is played on a graph $G$ called the \emph{sight graph}.
	\begin{enumerate}
		\item At every node $v\in V$ there is a player.
		\item Every player is assigned a hat colour $s_v$ from $\llbracket1,s\rrbracket=\{1,2,\ldots s\}$ uniformly randomly and independent of the hats of the other players.
		\item The directed edges of the graph show who can see whom, they are directed in the direction of information (if the vertex $y$ is visible from the vertex $x$, then the edge is $y\rightarrow x$), see Figure~\ref{fig:dir-arrows}.
		\item The players win as a team if every single one of them guesses their own hat colour correctly and lose otherwise.
	\end{enumerate}
	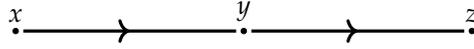
\begin{figure}[h]
		\centering
		\begin{tikzpicture}
			
			\filldraw (0,0) circle (1pt) node[align=left, above] {$x$};
			\filldraw (3,0) circle (1pt) node[align=center, above] {$y$};
			\filldraw (6,0) circle (1pt) node[align=right, above] {$z$};
			\begin{scope}[very thick,decoration={
					markings,
					mark=at position 0.5 with {\arrow{>}}}
				]
				\draw[postaction={decorate}] (0.1,0)->(2.9,0);
				\draw[postaction={decorate}] (3.1,0)->(5.9,0);
			\end{scope}
		\end{tikzpicture}
		\caption{On the graph above, $y$ can see $x$, and $y$ is seen by $z$.}\label{fig:dir-arrows}
	\end{figure}
	So the players do not compete but cooperate to win. Although the players cannot communicate after the game starts (they cannot even hear what another guesses), that is when the hat colours are determined; they can agree on a strategy beforehand and they a priori know $G$, $s$ and which player is on which vertex. This is shortly called a game on $(G,s)$ or a $(G,s)$-game. As $G$ is loopless, clearly there is no strategy to win with probability $1$. The aim is to maximize the probability of winning. Below we give a formalization of what a strategy is.
	\begin{definition}[Strategy]
		Let $G=(V,E)$ be the graph on which the game is played with $s$ colours.
		
		A \emph{guessing function} for the player on $x\in V$ is a mapping from $\llbracket1,s\rrbracket^{\{y\in V\mid (y,x)\in E\}}$ to $\llbracket1,s\rrbracket$. Intuitively, a guessing function is a table that shows what to guess for every possible configuration of what this player can see. Note that these configurations are equiprobable.
		
		A \emph{strategy} is a family of guessing functions $\mathcal{F}=(f_v)_{v\in V}$ for every vertex of the graph.
	\end{definition}
	Note that there are finitely many guessing functions for a player and thus finitely many strategies for a given $(G,s)$-game.
	
	In case we want to talk about \emph{random strategy}, we call strategy defined above \emph{deterministic} and define random strategy as probability distribution on deterministic strategies. However, no random strategy can do any better than the best deterministic strategy in terms of probability of winning. Indeed the probability of winning for a random strategy is a weighted average of the probabilities of winning of deterministic strategies. Therefore we only concentrate on deterministic strategies.
	
	\begin{definition}[Guessing Number]\label{defin:gn}
		The \emph{guessing number} of a game  measures the increase in the probability of correctly guessing the colours when playing with an optimal strategy (compared to a trivial strategy of choosing arbitrary colours as answers). Formally we denote it $gn(G,s)$:
		$$
		gn(G,s):=\max_{\mathcal{F}\text{ strategy}}\log_s\frac{\prob[\text{winning with }\mathcal{F}]}{s^{-|V|}}$$
		$$=\max_{\mathcal{F}\text{ str.}}\log_s\left|\{\text{winning config. in }\llbracket1,s\rrbracket^V\text{ for }\mathcal{F}\}\right|
		$$
	\end{definition}
	\begin{remark}
		Intuition behind this definition: the guessing number of the graph is $k$, if a best strategy gives the probability of winning that is $s^k$ times larger compared to the naive strategy where each player chooses an arbitrary colour, independently of what hats the neighbours receive. This value is the same as the logarithm on base $s$ of the cardinality of the largest set of configurations on which there is a winning strategy.
	\end{remark}
	\begin{remark}\label{rem:acyclicity}
		The guessing number of an acyclic graph is $0$ \cite[Lemma~3]{riis2006utilising}.
	\end{remark}
	\begin{remark}
		There is a generalization of this concept independent of $s$, known as \emph{asymptotic guessing number}. This is the limit of the guessing number as $s$ goes to infinity: $\lim_{s\to\infty}gn(G,s)$. The limit exists (thus the asymptotic guessing number is well defined) and it is an upper bound for the guessing number for any value of $s$.
	\end{remark}
	There is no known algorithm to compute these numbers for a given graph, and for some graphs only upper and lower bounds (that do not match each other) are known. 
	
	The lower bounds are proven using fractional clique cover of the graph and the upper bounds are proven using information inequalities see \cite{christofides2011guessing, baber2016graph}. For undirected graphs (graphs where $x\rightarrow y$ if and only if $y\rightarrow x$) with less than 10 vertices, the upper bounds given by Shannon-type inequalities match the lower bounds \cite{baber2016graph} thus the guessing numbers are known (at least up to floating point arithmetic errors).
	
	Let us give a very brief sketch of how to reduce the problem of finding an upper bound on the guessing number to a question about entropies and inequalities. Ultimately to give an upper bound for the guessing number is to give an upper bound for the number of configurations on which the players can win with a best strategy. Hence we can consider a particular probability distribution over the hat colours, where winning configurations are equiprobable and losing configurations have probability $0$. In this configuration the total entropy (in base $s$) of all the random variables $(X_v)_{v\in V}$ that correspond to the hat colours, is exactly the logarithm (in base $s$) of the cardinality of the set of winning configurations. Not that the following conditions are satisfied by these random variables:
	\begin{itemize}
	\item As each hat can get at most $s$ different values, $H_s(X_v)\le 1$,
	\item As each guess is correct under this distribution, $H_s(X_v|(X_u)_{u\rightarrow v})=0$.
	\end{itemize}	
	See Section~\ref{Ch:HGG} for more details.
	
	{\it Guessing Games And Network Coding:}\label{s:gg-nc}
	
	The problem of maximizing the probability of winning, seemingly a recreational mathematics problem, is directly related to the multiple unicast problem in network information theory. This connection, among other similar problems mentioned in \cite{riis2006utilising}, motivates the definition of guessing games and the study of their guessing numbers. In this article we do not use the correspondence between multiple unicast problems and hat guessing games, however, in what follows we explain the connection between them for the sake of self-containedness.
	
	\begin{definition}[Multiple Unicast Network Problem]
		Let us consider a directed acyclic graph $G=(V_t\sqcup V_r\sqcup V_d, E)$ of an information network with $n=|V_t|=|V_d|$ and $E\subset(V_t\times V_r)\sqcup(V_t\times V_d)\sqcup(V_r\times V_r)\sqcup(V_r\times V_d)$\footnote{$E$ is a proper subset because $G$ is acyclic, in particular the sub-graph of $G$ induced by $V_r$ is acyclic hence $E$ cannot contain whole $V_r\times V_r$}. The transmitters $x_1,x_2,\ldots x_n$ each want to send a message $a_i$ from the alphabet $A_s$ to their corresponding destinations $x'_1,x'_2,\ldots x'_n$, and there are $m=|V_r|$ routers $r_1,r_2,\ldots r_m$ in-between. let us note $\rightarrow(v):=\{u\in V\mid(u,v)\in E\}$ the set of vertices visible to $v$. Encoding functions $(f_v:A_s^{\rightarrow{v}}\to A_s)_{v\in V_r}$ for routers and decoding functions $(f_v:A_s^{\rightarrow{v}}\to A_s)_{v\in V_d}$ for destinations are called a \emph{protocol}. The task of deciding whether a protocol to ensure that every transmitted message is received by the corresponding destination exists for such a network and if exists, finding such a protocol is called \emph{multiple unicast network problem}.
		
		Note that each node $v\in V_t\sqcup V_r$ can send only one message and that all the nodes in $\leftarrow(v):=\{u\in V\mid(v,u)\in E\}$ receive the same message from $v$. Neither can $v$ send different messages to each of them nor can it send two consecutive messages $a_1,a_2\in A_s$ in a row. First the transmitters send their messages, then every node that receives all its due messages (one for each of its incoming edges) sends their message and so on until every edge has transmitted a message.
	\end{definition}
	
	\begin{theorem}[\cite{riis2006utilising}]\label{thm:hgg-nit-equ}
		Given a multiple unicast network problem with the network graph $G$ with $|V_t|=|V_d|=n$, $|V_r|=m$ and a transmission alphabet $A_s$, one can construct a graph $G'=(V',E')$ with $|V'|=n+m$ such that the following are equivalent:
		\begin{enumerate}
			\item There exists a solution to the multiple unicast network problem on $G$ with the message set $A_s$.
			\item The guessing number of $(G',s)$ is at least $n$.
			\item The guessing number of $(G',s)$ is exactly $n$.
		\end{enumerate}
		Moreover, the solutions of the guessing game and the network problem are in one to one correspondence.
	\end{theorem}
	
	Note that computing the \emph{guessing number} of a graph is not always easy.
	
	\section{Symmetries}\label{symmetry-section}
	
	Symmetries of the underlying structures (access structures for secret sharing or sight graphs for guessing games) can be exploited in the optimization problems to decrease the complexity of the problem. When combined with the copy lemma, the symmetry constraints force symmetric solutions to our linear programs without the symmetric applications of the copy lemma, thus the use of symmetries may improve the resulting optimal value. From another perspective, by removing the symmetric applications of the copy lemma, which are very costly in the complexity, of a linear program we may lose the optimal value but symmetry constraints may decrease the loss in the optimal value.
	
	\subsection{Symmetries and Linear Programming}
	
	Consider an optimization problem of the following form
	
	\begin{equation}\label{lp-sym}
		\begin{array}{c}
			\max f(v)\\
			\text{subject to:}\\
			v\in E\subset\overline{\Gamma_n^*}
		\end{array}
	\end{equation}
	
	where $f:\mathbb{R}^{\mathcal{P}(\llbracket1,n\rrbracket)\setminus\varnothing}\to\mathbb{R}$ is a linear form. Then we make the following simple observation.
	
	\begin{lemma}\label{l:prop-sym}
		Suppose
		\begin{itemize}
			\item $E$ is convex,
			\item there exists a group $G$ which acts on vectors $\mathbb{R}^{2^n-1}$ such that:
			\begin{itemize}
				\item $E$ is invariant under $G$, i.e. for all $u\in E,\sigma\in G$, we have $\sigma\cdot u\in E$,
				\item $f$ is invariant under $G$, i.e. for all $u\in E,\sigma\in G$, we have $f(u)=f(\sigma\cdot u)$.
			\end{itemize}
		\end{itemize}
		Then we have an optimal solution of \ref{lp-sym} invariant under $G$. 
	\end{lemma}
	\begin{proof}
		Let $v\in E$ be an optimal solution. Since $E$ is invariant under $G$, for all $\sigma\in G$, $\sigma\cdot v$ is also in $E$. Moreover since $E$ is convex, average of these vectors, namely $$v'=\frac1{|G|}\sum_{\sigma\in G}\sigma\cdot v$$ is also in $E$. By definition $v'$ is symmetric under $G$. By linearity of $f$ and then by invariance of $f$ $G$ we have the following:
		$$
		\begin{array}{rcl}
			f(v')&=&f(\frac1{|G|}\sum_{\sigma\in G}\sigma\cdot v)\\
			&=&\frac1{|G|}\sum_{\sigma\in G}f(\sigma\cdot v)\\
			&=&\frac1{|G|}\sum_{\sigma\in G}f(v)\\
			&=&f(v)
		\end{array}
		$$
		Therefore $v'$ too is an optimal solution of \ref{lp-sym}.
	\end{proof}
	
	We can apply this observation to linear programs for secret sharing and hat guessing games. The set $E\subset\overline{\Gamma_n^*}$ will correspond to particular restrictions that stems from the application (access structure for secret sharing and sight graph for guessing games). Since we know (from the lemma above) that there exists a symmetric solution, we can simply add symmetry constraints in the conditions of the linear program without any loss in the optimal value.
	
	\subsection{Symmetries Of Access Structures}\label{s:sym}
	
	To prove lower bounds on the information ratio of an access structure we reduce the problem of secret sharing to one of linear programming, as it was proposed in \cite{padro2013finding}.
	Let us sketch this simple reduction. For a fixed access structure $(\mathcal{A},\mathcal{B})$, we want to bound the information ratio of all secret sharing schemes $(S_0,\ldots S_n)$ realizing this access structure. The fact that $(S_0,\ldots S_n)$ is a secret sharing scheme realizing the access structure $({\cal A}, {\cal B})$ can be expressed as a family of linear constraints for the entropies of the involved variables, see \eqref{eq:secret-sharing}. The information ratio of a scheme is by definition $\max_i \frac{H(S_i)}{H(S_0)}$.	This objective function is not a linear combination of entropies. To overcome this obstacle, first of all, we introduce the normalization condition $H(S_0)=1$ and reduce the objective function to $\max_i H(S_i)$. This is still not a linear combination of entropies, so we have to introduce one more parameter (real variable) $x$ and add the constraints
	$$
	x\ge H(S_i),\ i=1,\ldots n.
	$$
	Now we can take the value of $x$ as the objective function, the minimal value of $x$ provides the maximum of $H(S_i)$.
	
	We can add to this optimization problem the constraints representing the linear inequalities that are valid for entropies of all random variables and, in particular, the random variables $(S_0,\ldots S_n)$.
	We combine this technique of \cite{padro2013finding} with symmetry conditions and then a series of applications of the copy lemma and then take a linear relaxation of this optimization problem to get a linear program as in \cite{gurpinar2019use}.  
	Let us formulate this general scheme as a proposition.
	
	\begin{proposition}\label{p:linear-program-for-ratio}
		Let $A$ be an access structure with $n$ participants and $(S_0,S_1,\ldots S_n)$ be a secret sharing scheme realizing it. We extend this distribution by adding $\ell$ random variables $S_{n+1},\ldots S_{n+\ell}$ using the copy lemma.
		The linear program described below provides a lower bound on the information ratio of $A$:
		
		$$
		\begin{array}{ll}
			&\min x\\
			&\text{subject to:}\\
			(i)&x\ge h_{S_i}\text{ for every } i\in\llbracket1,n\rrbracket\\
			(ii)&\text{the equalities \eqref{eq:secret-sharing} for the entropies of }S_0,\ldots S_n\\
			&\text{ which define the access structure }A\\
			(iii)&\text{classical information inequalities}\\&\text{for }S_0,\ldots S_{n+\ell}\\
			(iv)&h_{S_0}=1\text{ normalization}\\
			(v)&\text{the symmetry constraints on the variables }S_i,i\in\llbracket1,n\rrbracket\\
			&\text{under the symmetry group of the access structure }A\\	
			(vi)&\text{the equalities for entropies that define each of the random}\\
			&\text{variables }S_{n+1},\ldots S_{n+\ell}\text{ as a copy of other variables}\\
			&\text{(with smaller indices), obtained using the copy lemma}
		\end{array}
		$$
	\end{proposition}
	\begin{proof}
		We refer to the sets of constraints above as items (i), (ii), (iii), (iv), (v) and (vi). 
		
		It is clear (\cite{farras2018improving},\cite{gurpinar2019use}) that a linear program with items (i), (ii), (iv) with a subset of (iii) involving only $S_0,\ldots S_n$ gives a lower bound on the information ratio, since they make the linear relaxation of the following optimization problem on random variables.
		
		$$
		\begin{array}{l}
			\min x\\
			\text{subject to:}\\
			x=\max_i\frac{H(S_i)}{H(S_0)}\\
			S_0,\ldots S_n\text{ are random variables realizing }(\mathcal{A},\mathcal{B}) 
		\end{array}
		$$
		
		A relaxation of this optimization problem is as follows.
		
		$$
		\begin{array}{l}
			\min x\\
			\text{subject to:}\\
			x=\max_ih_{S_i}\\
			h_{S_0}=1\\
			(h_{S_I})_{\varnothing\neq I\subset\llbracket1,n\rrbracket}\in\overline{\Gamma_n^*}\\
			(h_{S_I})_{\varnothing\neq I\subset\llbracket1,n\rrbracket}\text{ satisfies the linear relaxation of }\ref{eq:secret-sharing}\\
			\text{ for }(\mathcal{A},\mathcal{B})
		\end{array}
		$$
		
		This optimization problem is not a linear program because to characterize $\overline{\Gamma_n^*}$ we would need infinitely many linear information inequalities.
		
		Since the set of constraints \ref{eq:secret-sharing} for the access structure $(\mathcal{A},\mathcal{B})$ is invariant under the symmetry group of $(\mathcal{A},\mathcal{B})$ by definition and that the other constraints and the objective function are also invariant, we can use Lemma~\ref{l:prop-sym}. Hence, we can add symmetry conditions to this optimization problem without changing its optimal value. Thus, we get the following optimization problem.
		
		$$
		\begin{array}{l}
			\min x\\
			\text{subject to:}
			x=\max_ih_{S_i}\\
			h_{S_0}=1\\
			(h_{S_I})_{\varnothing\neq I\subset\llbracket1,n\rrbracket}\in\overline{\Gamma_n^*}\\
			(h_{S_I})_{\varnothing\neq I\subset\llbracket1,n\rrbracket}\text{ satisfies the linear relaxation of }\ref{eq:secret-sharing}\\
			\text{ for }(\mathcal{A},\mathcal{B})\\
			\forall\sigma\in G, (h_{S_{\sigma\cdot I}})_{\varnothing\neq I\subset\llbracket1,n\rrbracket}=(h_{S_I})_{\varnothing\neq I\subset\llbracket1,n\rrbracket}\\
			\text{(symmetry constraints)}
		\end{array}
		$$
		
		Here $G$ is the symmetry group of $A$ and the action of $G$ on the almost entropic vector is by permuting its coordinates $\sigma\cdot(h_{S_I})_{\varnothing\neq I\in\llbracket1,n\rrbracket}=(h_{S_{\sigma\cdot I}})_{\varnothing\neq I\in\llbracket1,n\rrbracket}$. This justifies the addition of item~(iv) in the linear program.
		
		Now we can add copy lemma constraints (item~(v)) and extend the Shannon-type inequalities to all variables (full item~(iii)). Although these (in general any universally true information inequalities) are redundant for an optimization problem on almost entropic points, we can take a linear relaxation of this optimization problem and get the linear program in the statement of the proposition.
		
	\end{proof}
	
	See the appendix for the computation of the symmetry groups of the access structures we study in this work.
	
	\subsection{Symmetries of Sight Graphs}
	
	An analogous argument works to justify the use of symmetries for the linear programs used to get an upper bound on the asymptotic guessing number of a sight graph. We use Lemma~\ref{l:prop-sym} again.
	
	\begin{proposition}\label{p:linear-program-for-number}
		Let $G=(V,E)$ be a sight graph with $n$ vertices and $X_1,\ldots X_n$ the associated random variables as in \cite{christofides2011guessing}. We extend this distribution by adding $\ell$ random variables $X_{n+1},\ldots X_{n+\ell}$ using the copy lemma.
		The linear program described below provides a lower bound on the asymptotic guessing number of $G$:
		
		$$
		\begin{array}{ll}
			&\max h_{X_{\llbracket1,n\rrbracket}}\\
			&\text{subject to:}\\
			(i)&h_{X_i}\le1\text{ for every } i\in\llbracket1,n\rrbracket\\
			(ii)&h_{\{X_i,X_j|(j,i)\in E\}}-h_{\{X_j|(j,i)\in E\}}=0\text{ for each }i\in\llbracket1,n\rrbracket\\
			(iii)&\text{classical information inequalities}\\&\text{for }X_1,\ldots X_{n+\ell}\\
			(iv)&\text{the symmetry constraints on the variables }X_i,i\in\llbracket1,n\rrbracket\\
			&\text{under the symmetry group of the sight graph }G\\
			(v)&\text{the equalities for entropies that define each of the random}\\
			&\text{variables }X_{n+1},\ldots X_{n+\ell}\text{ as a copy of other variables}\\
			&\text{(with smaller indices), obtained using the copy lemma}\\
		\end{array}
		$$
	\end{proposition}
	\begin{remark}
		In \cite{baber2016graph} non-Shannon-type inequalities were added to such a linear program instead of item~(v) (the copy lemma constraints).
	\end{remark}
	\begin{proof}
		We refer to the sets of constraints above as items (i), (ii), (iii), (iv), (v) and as in the previous proposition.
		
		It is known (\cite{christofides2011guessing}) that a linear program with items (i), (ii) and the subset of (iii) for random variables $X_1,\ldots X_n$ gives an upper bound on the asymptotic guessing number, as these constraints make the linear relaxation of the optimization problem below on random variables. 
		
		Let us first justify the item (v) using Lemma~\ref{l:prop-sym}.
		
		$$
		\begin{array}{l}
			\max H(X_1,\ldots X_n)\\
			\text{subject to:}\\
			H(X_i)\le1\\
			H(X_i|(X_j)_{(j,i)\in E})=0\\
			X_1,\ldots X_n\text{ are random variables}
		\end{array}
		$$
		
		Again we take a (not linear) relaxation of this optimization problem, it is as follows.
		
		$$
		\begin{array}{l}
			\max h_{X_{\llbracket1,n\rrbracket}}\\
			\text{subject to:}\\
			h_{X_i}\le1\text{ for all }i\in\llbracket1,n\rrbracket\\
			h_{\{X_i,X_j|(j,i)\in E\}}-h_{\{X_j|(j,i)\in E\}}=0\text{ for each }i\in\llbracket1,n\rrbracket\\
			(h_{X_I})_{\varnothing\neq I\subset\llbracket1,n\rrbracket}\in\overline{\Gamma_n^*}
		\end{array}
		$$
		
		Since the set of constraints of the optimization problem above is invariant under the symmetry group $S(G)$ of the graph $G=(V,E)$ by definition and that the objective function is also invariant, we can use Lemma~\ref{l:prop-sym} and add symmetry conditions without changing the optimal value:
		
		$$
		\begin{array}{l}
			\max h_{X_{\llbracket1,n\rrbracket}}\\
			\text{subject to:}\\
			h_{X_i}\le1\text{ for all }i\in\llbracket1,n\rrbracket\\
			h_{\{X_i,X_j|(j,i)\in E\}}-h_{\{X_j|(j,i)\in E\}}=0\text{ for each }i\in\llbracket1,n\rrbracket\\
			(h_{X_I})_{\varnothing\neq I\subset\llbracket1,n\rrbracket}\in\overline{\Gamma_n^*}\\
			\forall\sigma\in S(G), (h_{S_{\sigma\cdot I}})_{\varnothing\neq I\subset\llbracket1,n\rrbracket}=(h_{S_I})_{\varnothing\neq I\subset\llbracket1,n\rrbracket}\\
			\text{(symmetry constraints)}
		\end{array}
		$$
		
		Here the action of $S(G)$ on the almost entropic vector is by permuting its coordinates $\sigma\cdot(h_{S_I})_{\varnothing\neq I\in\llbracket1,n\rrbracket}=(h_{S_{\sigma\cdot I}})_{\varnothing\neq I\in\llbracket1,n\rrbracket}$.
		
		Now we can introduce the constraints from the copy lemma application in item~(v) and add the Shannon-type inequalities for all variables (full item~(iii)). Although these again are redundant for an optimization problem on almost entropic points, we can then take a linear relaxation of this optimization problem and get the linear program in the statement of the proposition.
		
	\end{proof} 
	
	\section{Secret Sharing}\label{ch:ss}
	
	There is a large class of access structures called \emph{linear access structures} (also known as \emph{vector space access structure}) that are ideal. Let us give their definition.
	
	\begin{definition}
		An access strucure is called \emph{linear} if the secret and the participants $1,\ldots n$ can be associated respectively to some vectors $v_0, v_1, \ldots v_n$ in a vector space such that a coalition $I$ is
		\begin{itemize}
			\item an accepted coalition if $v_0\in Vect((v_i)_{i\in I})$, that is $v_0$ belongs to the linear subspace span by the set of vectors $V_i$ for $i\in I$,
			\item forbidden otherwise.
		\end{itemize}
	\end{definition}

	Note that, in particular, threshold access structures are linear: we can take a vector space of dimension equal to the threshold $t$. Then choose vectors associated to participants one by one such that any $t$ of them are independent. This can be done by choosing the field large enough\footnote{If the field has $k$ elements, the vector space has $k^t$ elements. Suppose we have chosen $n$ vectors so far, this can forbid no more than $\binom n{t-1}k^{t-1}$ choices for the next vector. Thus $k>\binom n{t-1}$ ensures us.}.
	
	There is a more general class of access structures, namely access structures derived from matroids. The notion of a matroid provides a more abstract (combinatorial) notion of `independence' inspired by the properties of linear independence in vector spaces (see \cite{oxley1992matroid} for a detailed introduction).
	
	\subsection{Access Structures From Matroids}\label{s:mat}
	
	Let us give a simple way to obtain some matroids: Let $E$ be a set, if we choose some vector space and a vector $v_e$ for each $e\in E$ and then declare a subset of $E$ to be \emph{independent} when the corresponding vectors are linearly independent, this gives us a matroid. Matroids that can be obtained this way are called \emph{linearly representable}.
	
	The matroids with a ground set $E$ of cardinality 7 or less, as well as those with a ground set of cardinality 8 and rank different than 4, are all known to be \emph{linearly representable} \cite{fournier1971representation}. There exists 940 non-isomorphic matroids of rank 4 on 8 points (see \cite{oxley1992matroid}). There exist matroids with a ground set of cardinality 8 and rank 4 that are not linearly representable, and the first known such example and the most famous among them is the V\'amos matroid (\cite{vamos1968representation}, \cite[Proposition~2.2.26]{oxley1992matroid}).
	\begin{definition}
		The figure~\ref{img-vamos} defines V\'amos matroid: any set of size 4 or less except the five circuits marked below (as both sides of an edge or as four points on a coloured surface), are independent. The remaining sets, that is the five sets of size 4 marked in the figure as well as sets of size larger than 4, are dependent.
	\begin{figure}
		\begin{tikzpicture}
			\node[circle, draw] at (0,2) (01){$0,1$};
			\node[circle, draw] at (-2,0) (23){$2,3$};
			\node[circle, draw] at (2,0) (45){$4,5$};
			\node[circle, draw] at (0,-2) (67){$6,7$};
			\draw (01)--(23);
			\draw (01)--(45);
			\draw (23)--(45);
			\draw (23)--(67);
			\draw (45)--(67);
		\end{tikzpicture}
		\begin{tikzpicture}
			\coordinate (c0) at (0,2);
			\coordinate (c2) at (-2,0);
			\coordinate (c4) at (2,0);
			\coordinate (c6) at (0,-2);
			\coordinate (c1) at (1.4,1.9);
			\coordinate (c3) at (-0.6,-0.1);
			\coordinate (c5) at (3.4,-0.1);
			\coordinate (c7) at (1.4,-2.1);
			\draw[fill=blue!90,opacity=0.33] (c0)--(c1)--(c3)--(c2)--cycle;
			\draw[fill=blue!90,opacity=0.33] (c0)--(c1)--(c5)--(c4)--cycle;
			\draw[fill=blue!90,opacity=0.33] (c2)--(c3)--(c5)--(c4)--cycle;
			\draw[fill=blue!90,opacity=0.33] (c2)--(c3)--(c7)--(c6)--cycle;
			\draw[fill=blue!90,opacity=0.33] (c4)--(c5)--(c7)--(c6)--cycle;
			\draw[thick] (c0)--(c1)--(c3)--(c2)--cycle;
			\draw[thick] (c0)--(c1)--(c5)--(c4)--cycle;
			\draw[thick] (c2)--(c3)--(c5)--(c4)--cycle;
			\draw[thick] (c2)--(c3)--(c7)--(c6)--cycle;
			\draw[thick] (c4)--(c5)--(c7)--(c6)--cycle;
			\node[label={[label distance=0.1cm]112.5:$0$}] at (0,2) (l0){};
			\node[label={[label distance=0.1cm]157.5:$2$}] at (-2,0) (l2){};
			\node[label={[label distance=0.1cm]157.5:$4$}] at (2,0) (l4){};
			\node[label={[label distance=0.1cm]247.5:$6$}] at (0,-2) (l6){};
			\node[label={[label distance=0.1cm]67.5:$1$}] at (1.4,1.9) (l1){};
			\node[label={[label distance=0.1cm]337.5:$3$}] at (-0.6,-0.1)(l3){};
			\node[label={[label distance=0.1cm]337.5:$5$}] at (3.4,-0.1) (l5){};
			\node[label={[label distance=0.1cm]292.5:$7$}] at (1.4,-2.1) (l7){};
		\end{tikzpicture}
		\caption{V\'amos matroid. Edges on the left and blue coloured surfaces on the right show the 4-element circuits (minimal non-independent sets).}\label{img-vamos}
	\end{figure}
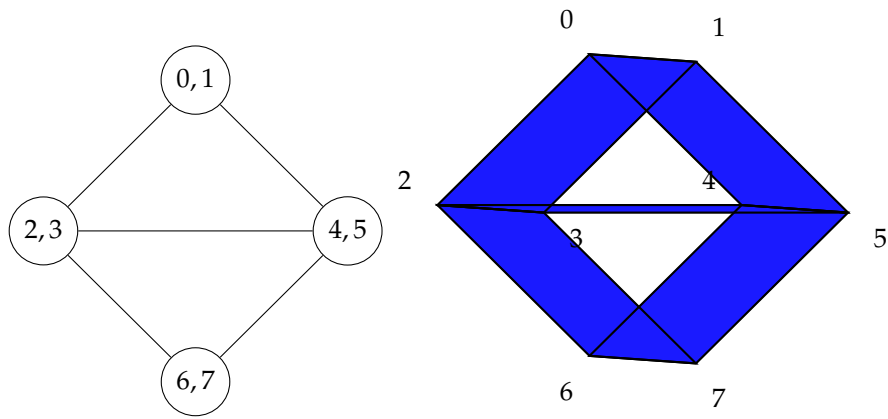
	\end{definition}
	
	Below we discuss seven other matroids with a ground set of cardinality 8 and rank 4 which are not linearly representable. 
	
	In the rest of the text we are only interested in connected matroids. More particularly \emph{matroid ports} (\cite{seymour1976forbidden}) and \emph{matroid port access structures} (\cite{farras2020secret}).
	
	\begin{remark}
		Since any upward closed set along with its complement define an access structure, so does any matroid port. Indeed, for any matroid $M'$ fix an element $p$ of the ground set $E$. We associate $p$ with the secret and identify the other elements of the ground set with the participants of a secret sharing scheme. We define the access structure in terms of the minimal accepted coalitions: a coalition $I$ is authorized if $I\cup\{p\}$ is a circuit.
		
		Note that the ground set of the matroid has $|E|$ elements, but the access structure we defined has $|E|-1$ participants.
	\end{remark}
	
	\begin{example}
	Let $\llbracket0,7\rrbracket$ be the ground set of the V\'amos matroid, as defined above.
	Let us associate $0$ with the secret and $\llbracket1,7\rrbracket$ with the participants of a secret sharing scheme. We obtain an access structure which has the following as the minimal accepted coalitions:
	\begin{itemize}
		\item $\{1,2,3\}$ and $\{1,4,5\}$\footnote{Because $\{0,1,2,3\}$ and $\{0,1,4,5\}$ are circuits.}
		\item all subsets of cardinality four of $\llbracket1,7\rrbracket$ except $\{2,3,4,5\}$, $\{2,3,6,7\}$, $\{4,5,6,7\}$, those containing $\{1,2,3\}$ and those containing $\{1,4,5\}$\footnote{Because all size five subsets of $\llbracket0,7\rrbracket$ containing $0$ are circuits except those that are strict supersets of circuits and that the circuits of size four or less are precisely $\{0,1,2,3\}$, $\{0,1,4,5\}$, $\{2,3,4,5\}$, $\{2,3,6,7\}$ and $\{4,5,6,7\}$.}
	\end{itemize}
	\end{example}

	In what follows we denote this access structure $\mathcal{V}$.
	
	\begin{remark}
		Note that if we associate the secret with the element $1$, $6$ or $7$ rather than $0$, then we obtain the same access structure up to renaming the participants. Indeed, these four elements in V\'amos matroid are interchangeable (symmetric to each other)\footnote{They form an orbit under the automorphism group $Aut(V_8)$ of the V\'amos matroid.}, as can be seen from Figure~\ref{img-vamos}. However, if we associate the secret with one of the elements $2$, $3$, $4$ or $5$, then we obtain a different access structure (denoted $\mathcal{V}^*$). The minimal accepted coalitions in this structure (up to renaming) are 
		\[
		\begin{array}{l}
			123,345,367,\text{ and all size four subsets of }\llbracket1,7\rrbracket\\
			\text{except those containing }123,345,367,1245,4567.
		\end{array}		
		\]
	\end{remark}
	
	It is known that ideal secret sharing is possible only if the access structure is a matroid port \cite{brickell1991classification}. Mart\'i-Farré and Padr\'o proved in \cite{marti2010secret} that any access structure with an information ratio less than $\sfrac32$ is based on a matroid. Thus, it is interesting to study the access structures in between these extremal cases: matroid ports without ideal secret sharing.
	
	It is clear by definition that the ports of linearly representable matroids are linear access structures and therefore, ideal access structures. So all the ports of matroids with seven or less points are known to have an ideal secret sharing. Thus, the problem of secret sharing on matroid ports is non-trivial for matroids with a ground set of cardinality eight or higher. The ports of matroids on eight points are access structures for seven participants (we denote them $1,2,3,4,5,6,7$, to be consistent with the
	
	\begin{table}[ht]
		\centering
		\begin{tabular}{|c|c|c|c|}
			\hline
			\thead{Access\\structure}&\thead{previously known lower\\bound based on\\AK lemma \cite{farras2020improving}}&\thead{bounds we prove\\using symmetries}&\thead{weaker bounds we can\\prove without symmetries}\\
			\hline
			$\mathcal{A}$&$\sfrac98=1.125$&$\sfrac{57}{50}=1.14$&$\sfrac{135}{119}=1.134\dots$\\
			\hline 
			$\mathcal{A}^*$&$\sfrac{33}{29}=1.137\dots$&$\sfrac{52}{45}=1.1\overline{5}$&$\sfrac{33}{29}=1.137\dots$\\
			\hline
			$\mathcal{F}$&$\sfrac98=1.125$&$\sfrac{17}{15}=1.1\overline{3}$&$\sfrac{26}{23}=1.130\dots$\\
			\hline
			$\mathcal{F}^*$&$\sfrac{42}{37}=1.\overline{135}$&$\sfrac87=1.142\dots$& $\sfrac{42}{37}=1.\overline{135}$\\
			\hline
			$\widehat{\mathcal{F}}$&$\sfrac{42}{37}=1.\overline{135}$&$\sfrac{23}{20}=1.15$&$\sfrac{42}{37}=1.\overline{135}$\\ 
			\hline
			$\mathcal{Q}$&$\sfrac98=1.125$&$\sfrac{17}{15}=1.1\overline{3}$&$\sfrac{17}{15}=1.1\overline{3}$\\ 
			\hline
			$\mathcal{Q}^*$&$\sfrac{33}{29}=1.137\dots$&$\sfrac87=1.142\dots$&$\sfrac{33}{29}=1.137\dots$\\
			\hline
		\end{tabular}	
		\smallskip
		\caption{The access structures of which we have improved lower bounds on the information ratio.}\label{table-sss-results}
	\end{table}
	
	\noindent
	notation from \cite{farras2020improving}).
	
	We focus on the access structure $\mathcal{V}$ and a few access structures whose study was initiated in \cite{farras2020improving}. All these access structures have some nice geometric interpretation, so do the matroids of which they are ports. 
	In \cite{farras2020improving}, they are named after the matroids $AG(3,2)',F_8$ and $Q_8$ from the appendix of \cite{oxley1992matroid}, from which they are derived. We follow their notation. As usual, each access structure can be defined by its minimal authorized coalitions, see Table~\ref{table:acc_str_stu}.
	
	\begin{table}[ht]
		\centering
		\begin{tabular}{ |c|c| }
			\hline
			\thead{Acccess\\Structure} & \thead{List of Minimal\\Authorized Sets} \\
			\hline
			$\mathcal{V}$ & \makecell{123, 145, 1246, 1247, 1256, 1257, 1267,\\1346, 1347, 1356, 1357, 1367, 1467,\\ 1567, 2346, 2347, 2356, 2357, 2456, 2457,\\2467, 2567, 3456, 3457, 3467, 3567} \\
			\hline
			$\mathcal{A}$ & \makecell{123, 145, 167, 246,\\ 257, 347, 356, 1247} \\
			\hline
			$\mathcal{A}^*$ & \makecell{123, 145, 167, 246, 257,\\ 347, 1356, 2356, 3456, 3567} \\
			\hline
			$\mathcal{F}$ & \makecell{123, 145, 167, 246, 257,\\ 347, 356, 1247, 1256} \\
			\hline
			$\mathcal{F}^*$ & \makecell{123, 145, 167, 246, 257, 1347, 1356,\\ 2347, 2356, 3456, 3457, 3467, 3567} \\
			\hline
			$\widehat{\mathcal{F}}$ & \makecell{123, 145, 167, 246, 257, 347,\\ 1256, 1356, 2356, 3456, 3567} \\
			\hline
			$\mathcal{Q}$ & \makecell{123, 145, 167, 246, 257, 347,\\ 1247, 1256, 1356, 2356, 3456, 3567} \\
			\hline
			$\mathcal{Q}^*$ & \makecell{123, 145, 167, 246, 257, 1247, 1347,\\ 1356, 2347, 2356, 3456, 3457, 3467, 3567} \\
			\hline
		\end{tabular}
		\caption{Access structures}\label{table:acc_str_stu}
	\end{table}
	
	The ultimate goal of this line of research is to find the information ratio for each of these access structures (and study the connection of information ratio with the combinatorial properties of matroids). This goal was not achieved in \cite{farras2020improving}, nor is it in our work. However, we take a new step in this direction and improve the known lower bound for the information ratio of these 8 access structures.

	\subsection{Results}\label{SSS:R}
	
	We improve the lower bounds for seven access structures.
	
	\begin{theorem}\label{thm:spanish-bounds}
		The information ratios of $\mathcal{A}$, $\mathcal{A}^*$, $\mathcal{F}$, $\mathcal{F}^*$, $\widehat{\mathcal{F}}$, $\mathcal{Q}$ and $\mathcal{Q}^*$ are $\sfrac{57}{50}=1.14$, $\sfrac{52}{45}=1.1\overline{5}$, $\sfrac{17}{15}=1.1\overline{3}$, $\sfrac87=1.142\dots$, $\sfrac{23}{20}=1.15$, $\sfrac{17}{15}=1.1\overline{3}$ and $\sfrac87=1.142\dots$ respectively. In the column 3 of Table~\ref{table-sss-results} we show these lower bounds for the information ratio of each access structure.
	\end{theorem}
	
	\begin{proof}
		For each of the seven access structures we construct a linear program as explained in Proposition~\ref{p:linear-program-for-ratio}. In this linear program we use auxiliary random variables with one or two applications of the copy lemma inspired from the applications of the AK lemma in \cite{farras2020improving}. We also add the constraints to express for each access structure the symmetry conditions (see Section~\ref{s:sym} and the appendix).
		
		We use the following applications of the copy lemma to create four additional variables for each access structure. To denote a copy of $X$, we use $X'$ in the first application of the copy lemma, $X''$ in the second copy step etc. and $X^{(i)}$ in the $i^{th}$:
		
		\begin{itemize}	
			\item $\mathcal{A}$: we introduce new variables $(S_0',S_3',S_4',S_7')$ as a copy of $(S_0,S_3,S_4,S_7)$.
			\item $\mathcal{A}^*$: we introduce $(S_0',S_3')$ as a $(S_5,S_6)$-copy of $(S_0,S_3)$ and then another pair $(S_1'',S_2'')$  as a $(S_0,S_0',S_3,S_3')$-copy of $(S_1,S_2)$.
			\item $\mathcal{F}$: we introduce new variables $(S_0',S_2',S_4',S_6')$ as a copy of $(S_0,S_2,S_4,S_6)$.
			\item $\mathcal{F}^*$: we introduce $(S_0',S_4')$ as a $(S_3,S_7)$-copy of $(S_0,S_4)$ and then $(S_1'',S_4'')$ as a $(S_0,S_0',S_4',S_5)$-copy of $(S_1,S_4)$.
			\item $\widehat{\mathcal{F}}$: we introduce $(S_0',S_4')$ as a $(S_2,S_6)$-copy of $(S_0,S_4)$ and then $(S_1'',S_4'')$ as a $(S_0,S_0',S_4',S_5)$-copy of $(S_1,S_4)$.
			\item $\mathcal{Q}$: we introduce $(T',V')$ as a $(S_0,S_2,S_4,S_6)$-copy of $(T,V)$ over $(S_1,S_3,S_5,S_7)$ and then $(T'',V'')$ as a  $(S_0,S_2,S_4,S_6,T',V')$-copy of $(T,V)$ over 
			$(S_1,S_3,S_5,S_7)$, where $T=(S_0,S_4)$ and $V=(S_2,S_6)$.
			\item $\mathcal{Q}^*$: we introduce a $(S_0',S_4')$ as a $(S_3,S_7)$-copy of $S_0,S_4$ and $(S_1'',S_5'')$ as a $(S_0,S_0',S_4,S_4')$-copy of $(S_1,S_5)$.
		\end{itemize}
		
		For comparison, in column 4 of Table~\ref{table-sss-results}, we show the weaker bound (strictly except for $Q$) that can be proven with the same use of the copy lemma but without symmetry conditions.
	\end{proof}
	
	\section{Hat-Guessing Games}\label{Ch:HGG}
	
	In this section we first describe in detail the lower and upper bound methods from \cite{christofides2011guessing} and \cite{baber2016graph} and then show our results. Our main result is on an undirected graph called $R^-$ with 10 vertices. The best known lower bound on its guessing number is $\sfrac{20}3=6.\bar{6}$ and the previously known best upper bound was $\sfrac{59767}{8929}=6.693\ldots$. We get an upper bound $\sfrac{1847}{276}=6.6920\ldots$ ($\approxeq6.692028986$).
	Our other results is to use our tools to get another proof of a known lower bound (for $R^L$) from \cite{baber2016graph}.
	
	\subsection{Asymptotic Guessing Number And Bounds Via Graphs}\label{s:asy-gue-num}
	
	In this section we briefly summarize the known methods of lower and upper bounds for the guessing number. The method suggested by Christofides and Markstr\"om in \cite{christofides2011guessing} to compute a lower bound for the guessing number is as follows. We first restrict ourselves to undirected graphs. (Note that we identify undirected edges and pairs of edges between the same vertices in both direction in directed graphs, hence we see undirected graphs as a subclass of directed graphs.) To state the technique, the following definition is necessary.
	
	\begin{definition}[Clique Cover Number, Definition~2.5 in \cite{christofides2011guessing}]
		A \emph{clique cover}, or \emph{clique partition} of an undirected graph $G$ with vertex set $V$ is a partition of $V$ into disjoint cliques. The \emph{clique cover number} $cp(G)$ of $G$ is the minimum number of (vertex disjoint) cliques into which $G$ can be partitioned.
	\end{definition}	
	
	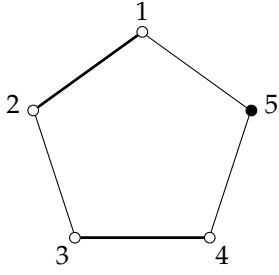
\begin{figure}[h]
		\centering
		\begin{tikzpicture}[mystyle/.style={draw,shape=circle,fill=white, inner sep=0pt, minimum size=4pt, label={[anchor=center, label distance=2mm](90+360/5*(#1-1)):#1}}]
			\node[draw, regular polygon,regular polygon sides=5,minimum size=3cm] (p) {};
			\foreach\x in {1,...,4}{
				\node[mystyle=\x] (p\x) at (p.corner \x){};	
			}
			\node[draw,shape=circle,fill=black, inner sep=0pt, minimum size=4pt, label={[anchor=center, label distance=2mm](90+360/5*4):5}] (p5) at (p.corner 5){};
			\draw[line width = 1pt] (p1) -- (p2);
			\draw[line width = 1pt] (p3) -- (p4);
		\end{tikzpicture}
		\caption{The clique cover number of $C_5$ is $3$.}\label{cp(C_5)}
	\end{figure}
	
	The following lemma is inspired from the observation of games on complete graphs. 
	
	\begin{lemma}[\cite{christofides2011guessing}, Lemma 2.6]\label{lem:cp}
		Let $G$ be a graph with $n$ vertices and $s>1$ integer. Then $gn(G,s)\ge n-cp(G)$.
	\end{lemma}
		
	Christofides and Markst\"om also show that $gn(G)\le n-\alpha(G)$ where $\alpha(G)$ denotes the size of the largest independent set of vertices of $G$. Indeed this is straightforward as independent set of vertices is acyclic. One can put all the hats other than the players' on the independent set first. This already fixes their guesses. So the guesses are independent of their hat colours, for which there are $s^{\alpha(G)}$ configurations. Hence the probability of success is bounded from above by $s^{-\alpha(G)}$.
	
	The upper bound via $\alpha(G)$ along with the lower bound via $cp(G)$ already settles down the guessing number for large classes of graphs, such as perfect graphs for which these two numbers are equal. The smallest non-perfect graphs is $C_5$, hence we illustrate it in Figure~\ref{cp(C_5)}.
	
	Another important result of \cite{christofides2011guessing} is the generalization of this technique by removing $s$ from the formulas and replacing clique cover by \emph{fractional clique cover} (see below).
	
	The guessing number as defined in Definition~\ref{defin:gn} may depend on the size $s$ of the colour alphabet $A_s$ (for example for $C_5$ it does) as shown in \cite{christofides2011guessing}. This motivates the next definition.
	
	\begin{definition}[Asymptotic Guessing Number, Theorem~3.6 and Definition~3.7 in \cite{christofides2011guessing}]
		The following limit exists.
		$$\lim_{s\to\infty}gn(G,s)$$
		It is called the \emph{asymptotic guessing number} of $G$ and noted $gn(G)$. In particular, it is equal to $\sup_{s\ge2}gn(G,s)$.
	\end{definition} 
	
	\begin{definition}
		Let $G=(V,E)$ be a graph and $t\ge 2$ be an integer. A \emph{$t$-blow up} of $G$ is a graph $(V',E')$ where $V'=V\times\llbracket1,t\rrbracket$ and $E'=\{((u,i),(v,j))\mid(u,v)\in E\text{ and }i,j\in\llbracket1,y\rrbracket\}$.
		The $t$-blow-up of $G$ is noted $G(t)$.
	\end{definition}
	
	\begin{lemma}[\cite{baber2016graph}, Lemma III.2]\label{lem:blow-up}
		Let $G=(V,E)$ be a graph and $t,s\ge2$ integers. Then $$t\cdot gn(G,s^t)=gn(G(t),s).$$ 
	\end{lemma}
	
	\begin{remark}
		In fact, in the proof they define a bijection between the strategies on $(G(t),s)$ and those on $(G,s^t)$.
		The bijection $\psi:\mathcal{F}\mapsto\psi(\mathcal{F})$ that it defines between strategies for two games is compatible with the bijection between the colour configurations. The strategy $\mathcal{F}$ wins on the configuration $(a_v)_{v\in V}$ of hat colours on $(G,s^t)$ if and only if $\psi(\mathcal{F})$ wins on the configuration $\phi((a_v)_{v\in V})$ on $(G(t),s)$.
		The first direction defines $\psi^{-1}$, and the other direction defines $\psi$. (It is not difficult to see that this function is a bijection, so $\psi^{-1}$ uniquely determines $\psi$.) 
	\end{remark}
	
	\begin{definition}[Fractional Clique Cover Number]
		Let $G=(V,E)$ be a graph and $K$ be the set of its cliques.
		A fractional clique cover is a weighting $w:K\to[0,1]$ of cliques such that, for every vertex $v\in V$, the sum of weights of the cliques it belongs to is $1$. Formally $$\forall v\in V,\sum_{\substack{k\in K\\v\in k}}w(k)=1.$$ Among all fractional clique covers, one that minimizes the sum of all weights defines the \emph{fractional clique partition number} of $G$:
		$$cp_f(G)=\min_{w\text{ fractional clique cover}}\sum_{k\in K}w(k)$$
	\end{definition}
	
	Note that the fractional clique cover number of a graph can be calculated using linear programming since all the conditions are expressible as linear inequalities as well as the sum of all the weights to minimize.
	
	\begin{figure}	
		\centering
		\begin{tikzpicture}[mystyle/.style={draw,shape=circle,fill=white, inner sep=0pt, minimum size=4pt, label={[anchor=center, label distance=2mm](90+360/5*(#1-1)):#1}}]
			\node[draw, regular polygon,regular polygon sides=5,minimum size=3cm, line width = 0.7pt] (p) {};
			\foreach\x in {1,...,5}{
				\node[mystyle=\x] (p\x) at (p.corner \x){};	
			}
		\end{tikzpicture}	
		\caption{The fractional clique cover number of $C_5$ is $2.5$}
	\end{figure}
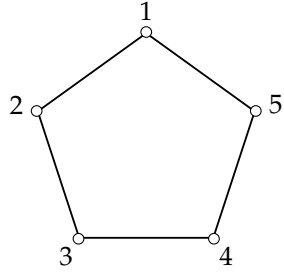
	It is known that the fractional clique cover number of every graph is rational.
	\begin{fact}
		For any graph $G$ and any integer $t\ge1$, we have $cp(G(t))\ge t\cdot cp_f(G)$.
	\end{fact}
	\begin{proof}
		Consider a clique cover of $G(t)$ with weights $w(k')$ for cliques $k'\in K(G(t))$. Identifying vertices $(v,i)_{i\in\llbracket1,t\rrbracket}$ of $G(t)$ with the vertex $v$ of $G$ partitions $K(G(t))$ into equivalence classes (the cliques $k'\in K(G(t))$ identified with the same clique $k\in K(G)$ are equivalent). This induces a fractional clique cover of $G$ by taking as weight of $k\in K(G)$ $$\frac1t\sum_{\substack{k'\in K(G(t))\\k'\text{ identified with }k}}w(k').$$
	\end{proof}
	\begin{fact}
		For every $G$, there exists an integer $t$ such that
		$cp(G(t))\le t\cdot cp_f(G)$.
	\end{fact}
	\begin{proof}
		Consider an optimal fractional clique cover of $G$, the sum of the weights of which is $cp_f(G)$. Express the weights $w(k),\ k\in K(G)$ using co-prime numerators $a_k$ and denominators $b_k$ (for each weight $w(k)=\frac{a_k}{b_k}$, $gcd(a_k,b_k)=1$). Let $t$ be the least common multiple of all these denominators $b_k,\ k\in K(G)$. Then we claim that $cp(G(t))\le t\cdot cp_f(G)$
		Let us now prove this inequality. By our choice of $t$, $t\cdot w(k)$ is integer for every $k\in K(G)$. For every clique $k\in G$ we can cover $t\cdot w(k)$ many cliques in $G(t)$ that are identified with $k$. For any clique of $G(t)$, the vertices $(v,i)$ are pairwise different in the second coordinate, and the clique cover defined this way will cover exactly $t$ vertices $(v,1),\ldots (v,t)$ of $G(t)$ for every vertex $v$ of $G$ as $\sum_{k\in K(G)}t\cdot w(k)=t$. Hence this clique cover is possible and optimal.
	\end{proof}
	
	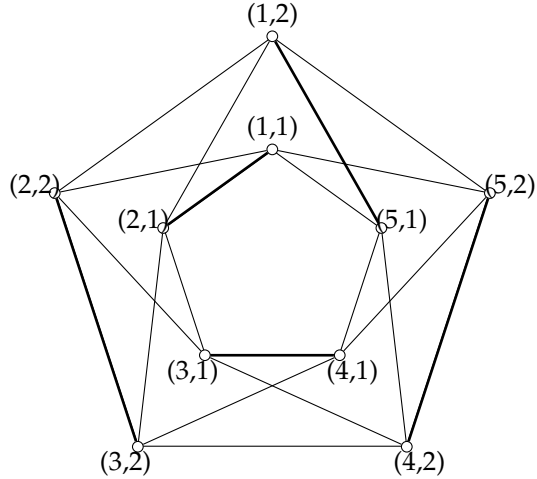
\begin{figure}	
		\centering
		\begin{tikzpicture}[mystyle/.style 2 args={draw,shape=circle,fill=white, inner sep=0pt, minimum size=4pt, label={[anchor=center, label distance=2mm](90+360/5*(#1-1)):{(#1,#2)}}}]
			\node[draw, regular polygon,regular polygon sides=5,minimum size=3cm] (p1) {};
			\node[draw, regular polygon,regular polygon sides=5,minimum size=6cm] (p2) {};	
			\foreach\x in {1,...,5}{
				\node[mystyle={\x}{1}] (p1\x) at (p1.corner \x){};	
				\node[mystyle={\x}{2}] (p2\x) at (p2.corner \x){};
			}
			\draw[-] (p11) -- (p22);
			\draw[-] (p12) -- (p23);
			\draw[-] (p13) -- (p24);
			\draw[-] (p14) -- (p25);
			\draw[-,line width=1pt] (p15) -- (p21);
			\draw[-] (p11) -- (p25);
			\draw[-] (p12) -- (p21);
			\draw[-] (p13) -- (p22);
			\draw[-] (p14) -- (p23);
			\draw[-] (p15) -- (p24);
			\draw[-,line width=1pt] (p11) -- (p12);
			\draw[-,line width=1pt] (p13) -- (p14);
			\draw[-,line width=1pt] (p22) -- (p23);
			\draw[-,line width=1pt] (p24) -- (p25);
		\end{tikzpicture}	
		\caption{The $2$-blow-up $C_5(2)$ of $C_5$ has clique cover number $5$.}
	\end{figure}
	
	Now we can state the main theorem of \cite{christofides2011guessing}.
	
	\begin{theorem}[\cite{christofides2011guessing},Theorem 4.9]\label{fractionallb}
		Let $G$ be a graph with n vertices. Let $t$ be such that $cp_f(G)=\sfrac{cp(G(t))}t$ and $s\ge2$ such that $s^{1/t}$ is integer. Then $gn(G,s)\ge n-cp_f(G)$ and thus the asymptotic guessing number has the following lower bound.
		$$gn(G)\ge n-cp_f(G)$$
	\end{theorem}
	
	\subsection{Upper Bounds On Guessing Number Via Entropy}\label{s:upp-bou-gue}
	
	In \cite{christofides2011guessing}, upper bounds on the guessing numbers of some graphs are proven with the help of Shannon-type information inequalities (the proofs are conventional, without use of computer). In \cite{baber2016graph} this method was extended: the authors explicitly used the formalism of linear programming and the assistance of a computer; these proofs involved non-Shannon-type information inequalities. In what follows we explain this technique.
	
	Consider a guessing game on $(G,s)$. Let us define jointly distributed random variables $(X_v)_{v\in V}$ associated with the vertices of the graph. Each random variable represents the hat colour of the player at vertex $v$. Let $\mathcal{F}$ be an optimal strategy on $(G,s)$. Instead of considering the independent uniformly random distribution for the colour of each hat, we consider the uniform distributions over all the configurations on which $\mathcal{F}$ wins. In other words, the colour configurations on which $\mathcal{F}$ loses all have probability $0$, and those on which $\mathcal{F}$ win are all equiprobable. Two things are special about this distribution. First of all, the entropy $H_s((X_v)_{v\in V})$ in base $s$ (using $\log_s$ instead of $\log_2$ in the definition of entropy) of all the variables is the logarithm of the cardinality of the set on which $\mathcal{F}$ wins, i.e. the guessing number by Definition~\ref{defin:gn}. Secondly, in this distributions the colours that are guessed are the same as the actual colours, hence the hat colour of a player is determined by the colours of the hats they see, therefore, $H(X_v|(X_u)_{u\in\rightarrow(v)})=0$\footnote{In the rest of the section, we do not use the index $s$ for entropy.}.
	
	Note that since the base of the logarithm $\log_s$ and the size of the alphabet $A_s$ (which is the image set of the random variables) are the same, the entropy of a random variable is bounded from above by $1$, i.e. $H(X_v)\le1,\ v\in V$.
	
	Now we can write an optimization problem and obtain via linear relaxation a linear program to upper bound the guessing number.
	
	\begin{proposition}[As discussed after Theorem~V.1 in \cite{baber2016graph}]\label{prop:lp-gg}
		Let $G$ be a graph, let us define a random variable $X_v$ for every vertex $v\in V$ as described above, then the optimization problem over random variables (and therefore, their entropies) below gives an upper bound on the guessing number $gn(G,s)$ for any $s\ge2$, hence for the asymptotic guessing number $gn(G)$.
		
		$$
		\begin{array}{l}
			\texttt{Maximize }\mathtt{H((X_v)_{v\in V})}\\
			\texttt{subject to:}\\
			\mathtt{H(X_v)\le 1}\\
			\mathtt{H(X_v|(X_u)_{u\in\rightarrow(v)})=0}
		\end{array}
		$$
		
		The linear program that we obtain by linear relaxation of this problem (we can add to the list of constraints of this linear program any universally true information inequalities for $(X_v)_{v\in V}$) also gives an upper bound on the asymptotic guessing number.
	\end{proposition}
	
	In \cite{christofides2011guessing} it was conjectured (Conjecture~6.4) that the asymptotic guessing number of a graph $G$ with $n$ vertices is always equal to $n-cp_f(G)$.\footnote{The authors use the notation $\chi_f(\overline{G})$ instead of $vp_f(G)$, since the chromatic number of the complement is clique cover number of $G$.} The authors of \cite{baber2016graph} wanted to check this conjecture for graphs with small number of vertices. They used the method of Proposition~\ref{prop:lp-gg} firstly with only Shannon-type inequalities and compared the upper bound given by this method to the lower bound given by the fractional clique cover on all undirected graphs with 9 or less vertices. They found that the bounds match (they performed the verification numerically, using  floating point arithmetic). On graphs with 10 vertices they found only 2 graphs (up to isomorphism) for which the lower and upper bounds do not match, called $R$ and $R^-$. The graph $R$ is given in Figure~\ref{graph-R} and $R^-$ is obtained by removing the edge $(9,10)$ from $R$.
	
	\begin{figure}[h]
		\centering
		\begin{tikzpicture}
			\node[draw, regular polygon,regular polygon sides=3,minimum size=3cm] (petit) {};
			\node[draw, white, regular polygon,regular polygon sides=3,minimum size=6cm] (grand) {};	
			\node[draw,shape=circle,fill=white, inner sep=0pt, minimum size=4pt, label={[anchor=center, label distance=2mm](90):2}] (p2) at (petit.corner 1){};	
			\node[draw,shape=circle,fill=white, inner sep=0pt, minimum size=4pt, label={[anchor=center, label distance=2mm](90-360/3):4}] (p4) at (petit.corner 3){};	
			\node[draw,shape=circle,fill=white, inner sep=0pt, minimum size=4pt, label={[anchor=center, label distance=2mm](90-360/3*2):6}] (p6) at (petit.corner 2){};	
			\node[draw,shape=circle,fill=white, inner sep=0pt, minimum size=4pt, label={[anchor=center, label distance=2mm](30):3}] (p3) at (petit.side 3){};	
			\node[draw,shape=circle,fill=white, inner sep=0pt, minimum size=4pt, label={[anchor=center, label distance=2mm](150):7}] (p7) at (petit.side 1){};	
			\node[draw,shape=circle,fill=white, inner sep=0pt, minimum size=4pt, label={[anchor=center, label distance=2mm](270):5}] (p5) at (petit.side 2){};	
			\node[draw,shape=circle,fill=white, inner sep=0pt, minimum size=4pt, label={[anchor=center, label distance=2mm](105):1}] (p1) at (petit.center){};	
			\node[draw,shape=circle,fill=white, inner sep=0pt, minimum size=4pt, label={[anchor=center, label distance=2mm](90):8}] (p8) at (grand.corner 1){};
			\node[draw,shape=circle,fill=white, inner sep=0pt, minimum size=4pt, label={[anchor=center, label distance=2mm](90-360/3):9}] (p9) at (grand.corner 3){};	
			\node[draw,shape=circle,fill=white, inner sep=0pt, minimum size=4pt, label={[anchor=center, label distance=2mm](90-360/3*2):10}] (p10) at (grand.corner 2){};	
			\draw[-] (p8) -- (p9);
			\draw[-] (p8) -- (p10);
			\draw[-] (p1) -- (p2);
			\draw[-] (p1) -- (p3);
			\draw[-] (p1) -- (p4);
			\draw[-] (p1) -- (p5);
			\draw[-] (p1) -- (p6);
			\draw[-] (p1) -- (p7);
			\draw[-] (p2) .. controls +(330:1) and +(120:1) .. (p9);
			\draw[-] (p2) .. controls +(210:1) and +(60:1) .. (p10);
			\draw[-] (p3) .. controls +(90:1) and +(300:1) .. (p8);
			\draw[-] (p3) .. controls +(330:1) and +(120:1) .. (p9);
			\draw[-] (p4) .. controls +(90:1) and +(300:1) .. (p8);
			\draw[-] (p4) .. controls +(210:1) and +(0:1) .. (p10);
			\draw[-] (p5) .. controls +(330:1) and +(180:1) .. (p9);
			\draw[-] (p5) .. controls +(210:1) and +(0:1) .. (p10);
			\draw[-] (p6) .. controls +(90:1) and +(240:1) .. (p8);
			\draw[-] (p6) .. controls +(330:1) and +(180:1) .. (p9);
			\draw[-] (p7) .. controls +(90:1) and +(240:1) .. (p8);
			\draw[-] (p7) .. controls +(210:1) and +(60:1) .. (p10);
		\end{tikzpicture}
		\caption{The graph $R^-$}\label{graph-R-}
	\end{figure}
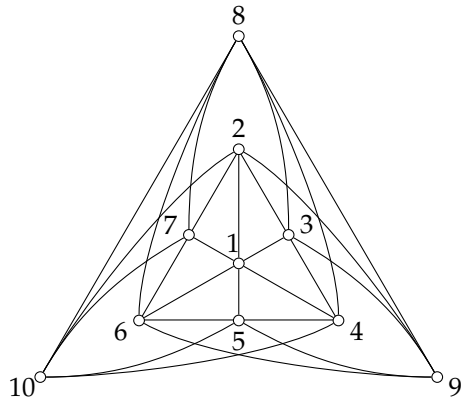
	
	The fractional clique cover number for $R$ and $R^-$ are both $\sfrac{10}3=3.\overline{3}$, which implies the lower bound $10-\sfrac{10}3=\sfrac{20}3=6.\overline{6}$ by Theorem~\ref{fractionallb}.
	
	The guessing number of $R$ is proven to be $\frac{27}4=6.75$ in \cite{baber2016graph} by an upper bound using Shannon-type inequalities and the construction of a strategy.
	
	The best upper bound for $R^-$ found in \cite{baber2016graph} using the non-Shannon-type inequalities from \cite{dougherty2011non} is $\sfrac{59767}{8929}=6.693\ldots$. 
	
	In \cite{baber2016graph}, the authors looked for an undirected graph such that the guessing number can be increased by adding one directed edge. They could not find such an example, and this motivated the question whether making a vertex `Superman' (visible by all others) by adding directed edges increases the guessing number. This led to the definition of the graph $R^S$ which is just as $R$ up to three outgoing edges from the vertex $1$ to vertices $8$, $9$ and $10$. The guessing number of $R^S$ is found to be $\sfrac{27}4-6.75$.
	
	Another question on guessing games on graphs: are there any graphs where the guessing number changes when the direction of all of its edges are reversed? This question has been motivated by the connection of guessing games with information networks (and the natural question of reversibility of networks). The authors of \cite{baber2016graph} looked at the candidates $R^S$ and its reverse $R^L$ in which $1$ is a `Luthor' vertex (sees all other vertices). A better lower bound for $R^L$ than its fractional clique cover number is given by the guessing number $\sfrac{27}4$ of $R$. The best upper bound they found on $R^L$ is $\sfrac{359}{53}=6.773\ldots$ using the non-Shannon-type inequalities from \cite{dougherty2011non}.
	
	\subsection{Our Results}\label{s:our-res}
	
	Combining the techniques of linear programming, the copy lemma and symmetries we improve the upper bound on $R^-$ and give an alternative proof of the previously known bound on $R^L$. For both of these graphs, the asymptotic guessing numbers remain unknown.
	
	The symmetry group of $R^-$ is $\langle\sigma=(18)(2\ 10\ 5\ 9)(3746),\tau=(25)(36)(47)\rangle$. The symmetry group of $R^L$ is $\langle(25)(36)(47),(26)(35)(8\ 10),(24)(57)(89)\rangle$. See the appendix for the proof that these are the maximal symmetry groups.
	
	Note that as in the case of secret sharing, we only need to find subgroups of these automorphism groups to apply our argument, and we do not need the whole group. We prove that those we use are indeed whole groups of automorphisms of these graphs to ensure that there is no obvious way to improve the bounds we get using them.
	
	Using these symmetries and the copy lemma we get the following linear programs and upper bounds for $gn(R^-)$ and therefore, improve the upper bound given in \cite{baber2016graph}.
	
	\begin{theorem}\label{th:R-}
		For the above defined graph $R^-$, $$gn(R^-)\le \sfrac{1847}{276}=6.6920\ldots.$$
	\end{theorem} 
	\begin{proof}
		We construct a linear program as in Proposition~\ref{prop:lp-gg} with the following constraints.
		\begin{enumerate}
			\item the following applications of the copy lemma:
			\begin{enumerate}
				\item \label{hgcl-i1} $X_2'$ be a $X_3$-copy of $X_2$;
				\item \label{hgcl-i2} \begin{itemize}
					\item $(X_4'',X_5'')$ be a $X_{10}$-copy of $(X_4,X_5)$ over $X_1,X_2,X_3,X_6,X_7,X_8,X_9$,
					\item and $X_7'''$ be a $(X_4'',X_5,X_{10})$-copy of $X_7$ over $X_1,X_2,X_3,X_4,X_5'',X_6,X_8,X_9$;
				\end{itemize}
				\item \label{hgcl-i3} \begin{itemize}
					\item $(X_6'''',X_7'''')$ be a copy of $(X_6,X_7)$ over $X_1,X_2,X_3,X_4,X_5,X_8,X_9,X_{10}$,
					\item and $X_8'''''$ be a $(X_7,X_7'''')$-copy of $X_8$ over $X_1,X_2,X_3,X_4,X_5,X_6,X_6'''',X_9,X_{10}$
				\end{itemize}
			\end{enumerate}
			\item the elemental inequalities for the following sets of random variables
			\begin{itemize}
				\item those that appear in the copy step in the item~\ref{hgcl-i1} above: 
				
				$X_1,X_2,X_2',X_3,X_4,X_5,X_6,X_7,X_8,X_9,X_{10}$
				\item those that appear in the copy steps of the item~\ref{hgcl-i2} above: 
				
				$X_1,X_2,X_3,X_4,X_4'',X_5,X_5'',X_6,X_7,X_7''',X_8,X_9,X_{10}$
				\item those that appear in the copy steps of the item~\ref{hgcl-i3} above: 
				
				$X_1,X_2,X_3,X_4,X_5,X_6,X_6'''',X_7,X_7'''',X_8,X_8''''',X_9,X_{10}$
			\end{itemize}
			\item the symmetry constraints for the group $\langle(18)(2\ 10\ 5\ 9)(3746),(25)(36)(47)\rangle$ 
		\end{enumerate}
		The optimal value of this linear program is $\le\sfrac{1847}{276}\approxeq6.692028986$, which proves the claim.
	\end{proof}
	
	We confirm the upper bound proven in \cite{baber2016graph}.
	\begin{theorem}
		For the graph $R^L$ defined above, $$gn(R^L)\le\sfrac{359}{53}=6.7735849\ldots.$$
	\end{theorem}
	\begin{proof}
		We construct a linear program as in the previous proof. We use the following constraints.
		\begin{enumerate}
			\item the following applications of the copy lemma:
			\begin{enumerate}
				\item \begin{itemize}
					\item $(X_4',X_5')$ be a copy of $(X_4,X_5)$,
					\item $X_5''$ be a $(X_1,X_4,X_4')$-copy of $X_5$,
					\item and $X_1'''$ be a $(X_4,X_4')$-copy of $X_1$;
				\end{itemize}  
				\item \begin{itemize}
					\item $(X_2'''',X_7'''')$ be a copy of $(X_2,X_7)$ over $X_1,X_3,X_4,X_5,X_6,X_8,X_9,X_{10}$,
					\item and $X_1'''''$ be a $(X_7,X_7'''')$-copy of $X_1$ over $X_2,X_2'''',X_3,X_4,X_5,X_6,X_8,X_9,X_{10}$
				\end{itemize} 
			\end{enumerate}
			\item the elemental inequalities for the following sets of random variables
			\begin{itemize}
				\item $X_1,X_1''',X_2,X_3,X_4,X_4',X_5,\\X_5',X_5'',X_6,X_7,X_8,X_9,X_{10}$
				\item $X_1,X_1''''',X_2,X_2'''',X_3,X_4,X_5,\\X_6,X_7,X_7'''',X_8,X_9,X_{10}$
			\end{itemize}
			\item the symmetry constraints for the permutations $(25)(36)(47)$, $(26)(35)(8\ 10)$ and $(24)(57)(89)$
		\end{enumerate}
		The optimal value of this linear program is $\le\sfrac{359}{53}\approxeq6.773584906$, which proves the claim.
	\end{proof}
	\begin{remark}
		Note that in the linear programs constructed in the proofs above, unlike those in secret sharing, we did not take all the elemental information inequalities for all the combinations of old and new random variables. For example there is no Shannon-type inequality involving both $X_2'$ and $X_4''$ in the first linear program and none involving both $X_4'$ and $X_2''''$ in the second. The reason is that the number of elemental inequalities involving all possible combinations of random variables is enormous. If we included all these constraints in the linear program, the computational complexity of the problem would increase so much that the existing linear program solvers (for our computers) could not handle it. Our choice of the sets of variables for which we write elemental information inequalities follows from the applications of the copy lemma: a copy variable used in order to define another copy variable is put in the same set as the latter.
	\end{remark}
	
	\section{Conclusion}
	In this paper we studied the application of computer-assisted proofs involving non-Shannon-type inequalities. Though each separate ingredient used in our construction was known earlier, the resulting combination proved to be surprisingly efficient.
	
	We improved lower bounds for the information ratio of 8 access structures based on rank-4 8-point not-linearly-representable matroids. We tried to apply the same approach to one more similar access structure (based on the matroid port $\mathcal{V}^*$), however we failed. We believe this `success rate' (8 bounds improved out of 9 access structures investigated) shows that the used method is quite strong and it might be interesting to extend to the other instances of the problem of secret sharing. 
	
	We also improved the upper bound for the single smallest undirected graph, the asymptotic guessing number of which is unknown, namely $R^-$. Not only our bound improves upon the previous one, but also the fraction is simpler (i.e. the denominator is smaller). Note that there is no evidence that the obtained number is the exact guessing number for this graph, a finer analysis may improve our upper bound.
	
	\section{Acknowledgement}
	
	I thank Andrei Romashchenko for his reading, suggestions, feedback and motivating me to finish this article. This work was partly funded by FLITTLA project (ANR-21-CE48-0023).
	
	\printbibliography

@article{zhang1998characterization,
	title={On characterization of entropy function via information inequalities},
	author={Zhang, Zhen and Yeung, Raymond Wai-Ho},
	journal={IEEE Transactions on Information Theory},
	volume={44},
	number={4},
	pages={1440--1452},
	year={1998},
	publisher={IEEE}
}

@book{yeung2002first,
  title={A first course in information theory},
  author={Yeung, Raymond Wai-Ho},
  year={2002},
  publisher={Springer Science \& Business Media}
}

@article{dougherty2011non,
	title={Non-Shannon information inequalities in four random variables},
	author={Dougherty, Randall and Freiling, Chris and Zeger, Kenneth},
	journal={arXiv preprint arXiv:1104.3602},
	year={2011}
}

@inproceedings{gurpinar2019use,
	title={How to use undiscovered information inequalities: Direct applications of the copy lemma},
	author={G{\"u}rp{\i}nar, Emirhan and Romashchenko, Andrei},
	booktitle={2019 IEEE International Symposium on Information Theory (ISIT)},
	pages={1377--1381},
	year={2019},
	organization={IEEE}
}

@article{padro2013finding,
	title={Finding lower bounds on the complexity of secret sharing schemes by linear programming},
	author={Padr{\'o}, Carles and V{\'a}zquez, Leonor and Yang, An},
	journal={Discrete applied mathematics},
	volume={161},
	number={7-8},
	pages={1072--1084},
	year={2013},
	publisher={Elsevier}
}

@article{shannon1948mathematical,
	title={A mathematical theory of communication},
	author={Shannon, Claude Elwood},
	journal={The Bell system technical journal},
	volume={27},
	number={3},
	pages={379--423},
	year={1948},
	publisher={Nokia Bell Labs}
}

@article{riis2007information,
	title={Information flows, graphs and their guessing numbers},
	author={Riis, S{\o}ren},
	journal={the electronic journal of combinatorics},
	pages={R44--R44},
	year={2007}
}

@article{christofides2011guessing,
	title={The guessing number of undirected graphs},
	author={Christofides, Demetres and Markstr\"om, Klas},
	journal={the electronic journal of combinatorics},
	pages={P192--P192},
	year={2011}
}

@article{baber2016graph,
	title={Graph guessing games and non-Shannon information inequalities},
	author={Baber, Rahil and Christofides, Demetres and Dang, Anh N and Vaughan, Emil R and Riis, S{\o}ren},
	journal={IEEE Transactions on Information Theory},
	volume={63},
	number={7},
	pages={4257--4267},
	year={2016},
	publisher={IEEE}
}

@inproceedings{blakley1979safeguarding,
	title={Safeguarding cryptographic keys},
	author={Blakley, George Robert},
	booktitle={Managing Requirements Knowledge, International Workshop on},
	pages={313--313},
	year={1979},
	organization={IEEE Computer Society}
}

@article{shamir1979share,
	title={How to share a secret},
	author={Shamir, Adi},
	journal={Communications of the ACM},
	volume={22},
	number={11},
	pages={612--613},
	year={1979},
	publisher={ACm New York, NY, USA}
}

@article{padro2012lecture,
	title={Lecture notes in secret sharing},
	author={Padr{\'o}, Carles},
	journal={Cryptology ePrint Archive},
	year={2012}
}

@article{zhang1997non,
	title={A non-Shannon-type conditional inequality of information quantities},
	author={Zhang, Zhen and Yeung, Raymond Wai-Ho},
	journal={IEEE Transactions on Information Theory},
	volume={43},
	number={6},
	pages={1982--1986},
	year={1997},
	publisher={IEEE}
}

@article{ito1989secret,
	title={Secret sharing scheme realizing general access structure},
	author={Ito, Mitsuru and Saito, Akira and Nishizeki, Takao},
	journal={Electronics and Communications in Japan (Part III: Fundamental Electronic Science)},
	volume={72},
	number={9},
	pages={56--64},
	year={1989},
	publisher={Wiley Online Library},
	note="Japanese publication 1988"
}

@article{csirmaz1997size,
	title={The size of a share must be large},
	author={Csirmaz, L{\'a}szl{\'o}},
	journal={Journal of cryptology},
	volume={10},
	number={4},
	pages={223--231},
	year={1997},
	publisher={Springer}
}

@article{capocelli1993size,
	title={On the size of shares for secret sharing schemes},
	author={Capocelli, Renato M. and De Santis, Alfredo and Gargano, Luisa and Vaccaro, Ugo},
	journal={Journal of Cryptology},
	volume={6},
	number={3},
	pages={157--167},
	year={1993},
	publisher={Springer}
}

@inproceedings{beimel2008matroids,
	title={Matroids can be far from ideal secret sharing},
	author={Beimel, Amos and Livne, Noam and Padr{\'o}, Carles},
	booktitle={Theory of Cryptography Conference},
	pages={194--212},
	year={2008},
	organization={Springer}
}

@article{metcalf2011improved,
	title={Improved upper bounds for the information rates of the secret sharing schemes induced by the V{\'a}mos matroid},
	author={Metcalf-Burton, Jessica Ruth},
	journal={Discrete Mathematics},
	volume={311},
	number={8-9},
	pages={651--662},
	year={2011},
	publisher={Elsevier}
}

@article{farras2020improving,
	title={Improving the linear programming technique in the search for lower bounds in secret sharing},
	author={Farr{\`a}s, Oriol and Kaced, Tarik and Mart{\'i}n, Sebasti{\`a} and Padr{\'o}, Carles},
	journal={IEEE Transactions on Information Theory},
	volume={66},
	number={11},
	pages={7088--7100},
	year={2020},
	publisher={IEEE}
}

@article{bamiloshin2021common,
	title={Common information, matroid representation, and secret sharing for matroid ports},
	author={Bamiloshin, Michael and Ben-Efraim, Aner and Farr{\`a}s, Oriol and Padr{\'o}, Carles},
	journal={Designs, Codes and Cryptography},
	volume={89},
	number={1},
	pages={143--166},
	year={2021},
	publisher={Springer}
}

@article{marti2010secret,
	title={On secret sharing schemes, matroids and polymatroids},
	author={Mart{\'i}-Farr{\'e}, Jaume and Padr{\'o}, Carles},
	journal={Journal of Mathematical Cryptology},
	volume={4},
	number={2},
	pages={95--120},
	year={2010},
	publisher={De Gruyter}
}

@book{oxley1992matroid,
	title={Matroid theory},
	author={Oxley, James},
	year={2011},
	note={Second edition},
	publisher={Oxford University Press}
}

@article{brickell1991classification,
	title={On the classification of ideal secret sharing schemes},
	author={Brickell, Ernest Francis and Davenport, Daniel M.},
	journal={Journal of Cryptology},
	volume={4},
	number={2},
	pages={123--134},
	year={1991},
	publisher={Springer}
}

@inproceedings{benaloh1988generalized,
	title={Generalized secret sharing and monotone functions},
	author={Benaloh, Josh and Leichter, Jerry},
	booktitle={Conference on the Theory and Application of Cryptography},
	pages={27--35},
	year={1988},
	organization={Springer}
}

@inproceedings{beimel2011secret,
	title={Secret-sharing schemes: A survey},
	author={Beimel, Amos},
	booktitle={International conference on coding and cryptology},
	pages={11--46},
	year={2011},
	organization={Springer}
}

@incollection{fournier1971representation,
	title={Representation sur un Corps des Matroides d’Ordre$\le$8},
	author={Fournier, J. C.},
	booktitle={Th{\'e}orie des Matro{\"\i}des},
	pages={50--61},
	year={1971},
	publisher={Springer}
}

@article{vamos1968representation,
	title={On the representation of independence structures},
	author={Vamos, Peter},
	journal={Unpublished manuscript},
	year={1968}
}

@article{seymour1976forbidden,
	title={A forbidden minor characterization of matroid ports},
	author={Seymour, Paul D.},
	journal={The Quarterly Journal of Mathematics},
	volume={27},
	number={4},
	pages={407--413},
	year={1976},
	publisher={Oxford University Press}
}

@article{farras2020secret,
	title={Secret sharing schemes for ports of matroids of rank 3},
	author={Farr{\`a}s, Oriol},
	journal={Kybernetika},
	volume={56},
	number={5},
	pages={903--915},
	year={2020},
	publisher={Institute of Information Theory and Automation AS CR}
}

@article{riis2006utilising,
	title={Utilising public information in network coding},
	author={Riis, S{\o}ren},
	journal={General Theory of Information Transfer and Combinatorics},
	volume={4123},
	pages={866--897},
	year={2006},
	publisher={Citeseer}
}

@inproceedings{farras2018improving,
	title={Improving the linear programming technique in the search for lower bounds in secret sharing},
	author={Farr{\`a}s, Oriol and Kaced, Tarik and Mart{\'\i}n, Sebasti{\`a} and Padr{\'o}, Carles},
	booktitle={Annual International Conference on the Theory and Applications of Cryptographic Techniques},
	pages={597--621},
	year={2018},
	organization={Springer}
}

@article{robinson2001mathematicians,
	title={Why mathematicians now care about their hat color},
	author={Robinson, Sara},
	journal={The New York Times, Science Times Section, page D},
	volume={5},
	year={2001}
}

@article{winkler2002games,
	title={Games people don't play},
	author={Winkler, Peter},
	year={2002},
	publisher={Citeseer}
}

@article{butler2009hat,
	title={Hat guessing games},
	author={Butler, Steve and Hajiaghayi, Mohammad T and Kleinberg, Robert D and Leighton, Tom},
	journal={SIAM review},
	volume={51},
	number={2},
	pages={399--413},
	year={2009},
	publisher={SIAM}
}

@inproceedings{ma2011new,
	title={A new variation of hat guessing games},
	author={Ma, Tengyu and Sun, Xiaoming and Yu, Huacheng},
	booktitle={International Computing and Combinatorics Conference},
	pages={616--626},
	year={2011},
	organization={Springer}
}
	
	\section*{Appendices}
	\addcontentsline{toc}{section}{Appendices}
	\renewcommand{\thesubsection}{\Alph{subsection}}
	
	\subsection{Symmetry Groups of Access Structures}
	To find the symmetry constraints to be written in item~(vi) of the Proposition~\ref{p:linear-program-for-ratio}, we need to understand the symmetries of the access structures under consideration.
	So let us more closely look at the symmetry groups of the access structures we previously described.
	In the rest of this section $G$ denotes the largest permutation group under which the access structure in question is invariant.
	
	Below we find the symmetry groups.	
	
	\begin{itemize}
		\item $\mathcal{V}$: $\langle(24)(35),(23),(45),(67)\rangle$
		\item $\mathcal{A}, \mathcal{A}^*$: $\langle(12)(56),(14)(36),(17)(35)\rangle$
		\item $\mathcal{F}, \mathcal{F}^*$: $\langle(12)(4576),(46)(57)\rangle$
		\item $\mathcal{Q}, \widehat{\mathcal{F}}, \mathcal{Q}^*$: $\langle(12)(47),(12)(56)\rangle$
	\end{itemize}
	
	It is easy to check that these access structures are invariant under these permutations. In the following paragraphs we explain how to find for each access structure its group of symmetries.
	
	{\small
		\paragraph{For $\mathcal{V}$,} the symmetry group $G$ of the access structure (the stabilizer of $0$ in the automorphism group $Aut(V_8)$ of the V\'amos matroid) is generated by the permutations  $(23),(45),(67),(24)(35)$, as can be seen from Figure~\ref{img-vamos}. Indeed, 
		it is easy to verify that the access structure is $G$-invariant. Let us explain why $\mathcal{V}$ has no other symmetries (besides the elements of $G$).
		First notice that the sets $\{0,1\},\{2,3\},\{4,5\},\{6,7\}$ form a block system. This can be seen by looking at which of the five size four circuits they are in.
		Every automorphism that has $0$ as a fixed point, must also have $1$ as a fixed point since $\{0,1\}$ is a block (of size two).
		Further notice that the number of size four circuits these blocks are in: $2$, $3$, $3$ and $2$ respectively. Therefore, since the block $\{0,1\}$ is fixed, the block $\{6,7\}$ must also be fixed. This leaves only $\langle(67)\rangle$ as permutations for $6$ and $7$. 
		There remains only two blocks: $\{2,3\}$ and $\{4,5\}$. If one is fixed, so is the other; if they are not fixed, they must move to each other. In the former case we are left with the permutations $\langle(23),(45)\rangle$ and in the latter $(24)(35)\cdot\langle(23),(45)\rangle$, hence in general case $\langle(23),(45),(24)(35)\rangle$.
		
		\paragraph{For $\mathcal{A}$,} we present the access structure as in the image below:
		
		
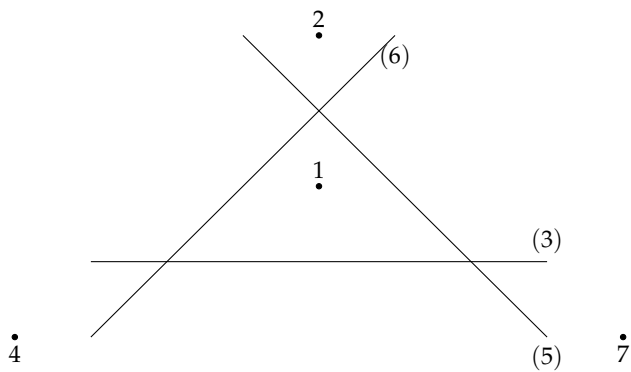
\begin{tikzpicture}
			\filldraw (0,0) circle (1pt) node[align=center, above] {$1$};
			\filldraw (0,2) circle (1pt) node[align=center, above] {$2$};
			\draw (-3,-1) -- (3,-1) node[align=right, above]{$(3)$};
			\filldraw (-4,-2) circle (1pt) node[align=left, below]{$4$};
			\draw (-1,2) -- (3,-2) node[align=right,below]{$(5)$};
			\draw (-3,-2) -- (1,2) node[align=right,below]{$(6)$};
			\filldraw (4,-2) circle (1pt) node[align=right,below]{$7$};
		\end{tikzpicture}
		
		\smallskip
		
		Any two points with the line according to which they are at the same side (for example $\{1,2,3\}$, because $3$ divides the plane into two half-planes, and $1$ and $2$ are on the same half-plane) is a minimum accepted coalition as well as all four points. In fact, the points are precisely those that appear in the only size four accepted coalition $1,2,4,7$, and the lines are those that do not. Hence the action of $G$ on the block system $\{\{1,2,4,7\},\{3,5,6\}\}$ is trivial (both blocks are fixed as their cardinalities are different). Thus $G\le Sym(\{1,2,4,7\})\cdot Sym(\{3,5,6\})$.
		
		We intuit that the action of $G$ on the set of `points' $\{1,2,4,7\}$ is transitive. The transpositions $(12)$, $(14)$ and $(17)$ generate the group $Sym(\{1,2,3,4\})$. If for each of these transpositions $\tau$ we can find a non-empty subset $S_\tau\subseteq Sym(\{3,5,6\})$ under the multiplication  with which $\mathcal{A}$ is invariant, then we can conclude that $G$ is generated by the elements of three sets $\tau\cdot S_\tau,\ \tau=(12),(14),(17)$. A priori, if the action of $G$ on $\{1,2,4,7\}$ is not transitive, the sets $S_\tau$ are not all non-empty. If at least one of them is empty, we cannot characterize $G$ with the others. However, it is not hard to verify that $S_{(12)}=\{(56)\}$, $S_{(14)}=\{(36)\}$ and $S_{(17)}=\{(35)\}$. This proves our intuition and that $G$ is $\langle(12)(56),(14)(36),(17)(35)\rangle$.
		
		\paragraph{The same argument works for $\mathcal{A}^*$,} we just change the geometric definition of minimum authorized sets to `any two points with the line according to which they are at the same side as well as any point with all three lines' (for example $\{1,3,5,6\}$ because $1$ is a point and $3,5,6$ are the three lines).
		
		\paragraph{For $\mathcal{F}$,} we use the following geometric presentation:
		
		\smallskip
		
		
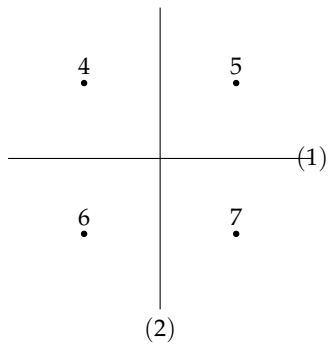
\begin{tikzpicture}
			\draw (-2,0) -- (2,0) node[align=right] {$(1)$};
			\draw (0,2) -- (0,-2) node[align=center, below] {$(2)$};
			\filldraw (-1,1) circle (1pt) node[align=left,above] {$4$};
			\filldraw (1,1) circle (1pt) node[align=left,above] {$5$};
			\filldraw (-1,-1) circle (1pt) node[align=left,above] {$6$};
			\filldraw (1,-1) circle (1pt) node[align=left,above] {$7$};
		\end{tikzpicture}
		
		\smallskip
		
		In fact, $3$ is the only one not appearing in the minimum authorized coalitions of size four, thus it must be fixed under $G$. $1$ and $2$ are the only ones appearing in both of the four-element accepted coalitions. Hence, we have a block system ($\{3\}$, $\{1,2\}$, $\{4,5,6,7\}$) on which the action of $G$ is trivial. Looking at the minimum authorized coalitions of size three, we get the rest of the image. It is easy to verify that the authorized sets are:
		\begin{itemize}
			\item with $3$:
			\begin{itemize}
				\item both lines
				\item two points separated by both lines (for example $\{3,4,7\}$ because $4$ and $7$ are separated by both lines)
			\end{itemize}
			\item without $3$:
			\begin{itemize}
				\item two points with a line not separating them (for example $\{1,4,5\}$ because $4$ and $5$ are on the same side of $1$)
				\item both lines with two points separated by both lines (for example $\{1,2,5,6\}$ because $5$ and $6$ are separated by both lines)
			\end{itemize}
		\end{itemize}
		The symmetries are mirror images and rotations for this image, therefore, they are generated by $(12)(4576)$ and $(46)(57)$.
		
		\paragraph{The same argument works for $\mathcal{F}^*$,} we only change the interpretation of the image to define the minimum authorized coalitions as
		\begin{itemize}
			\item with $3$:
			\begin{itemize}
				\item both lines
				\item a line with two points separated by both lines (for example $\{1,3,4,7\}$ because $4$ and $7$ are separated by both lines and that $1$ is a line)
				\item any three points (for example $\{3,4,5,6\}$, because $3,$, $5$ and $6$ are points)
			\end{itemize}
			\item without $3$: two points with a line not separating them'' (for example $\{1,4,5\}$ because $4$ and $5$ are on the same side of $1$).
		\end{itemize}
		
		\paragraph{For $\mathcal{Q}$, $\widehat{\mathcal{F}}$ and $\mathcal{Q}^*$,} we consider the following block system: $\{1,2\}$, $\{3\}$, $\{4,7\}$ and $\{5,6\}$.
		
		For $\mathcal{Q}$, the number minimum authorized coalition of size four to which a participant belongs reveal why the action of $G$ on the blocks is trivial: $1$ and $2$ each in three minimal accepted coalitions, $3$ in four, $4$ and $7$ two, $5$ and $6$ five.
		
		For $\widehat{\mathcal{F}}$, $3$ is the only one that belongs to four minimal authorized coalitions of size four, $1$ and $2$ each belong to two, $5$ and $6$ five and $4$ and $7$ one.
		
		For $\mathcal{Q}^*$, $3$ is the only one that belongs to a single minimal accepted coalition of size three. $1$ and $2$ are the only ones that are with $3$ in this coalition. $4$ and $7$ are the only ones to that belong to a single minimal authorized coalition of size four that contain both $1$ and $2$.
		
		Hence the symmetry group $G$ is $\langle(12)(56),(12)(47)\rangle$. 
		
		\begin{remark}
			In what follows we discuss computer-assisted proofs involving symmetry constraints.
			In practice, one can make an error while translating the description of a symmetry group in a form suitable for a computer. However, there is a fortunate `sanity check' embedded in this method.
			If there is an error in the symmetry group added as conditions to the linear program (if we use a wrong group $T$ instead of the symmetry group $G$ of our access structure),
			then two cases are possible:
			
			\begin{itemize}
				\item either $T<G$, then we still get a valid lower bound, possibly worse than what we could achieve with the true group of symmetries of this access structure;
				\item or $T\nless G$, then we get an infeasible program (with no solution). Indeed, any element of $T\setminus G$ applied to the equalities in item (iii) (see Proposition~\ref{p:linear-program-for-ratio}, page~\pageref{p:linear-program-for-ratio}) gives a contradiction. Namely, if $A$ is an accepted set but $\sigma\cdot A$ for a $\sigma\in T\nless G$ is not (there is such an $A$ since $\sigma\notin G$), then as $H(S_0|S_A)=0$, the symmetry under $\sigma$ implies $H(S_0|S_{\sigma\cdot A})=0$ too, but we have $H(S_0|S_{\sigma\cdot A})=H(S_0)$ and $H(S_0)=1$ as constraints, contradiction! 
			\end{itemize}
		\end{remark}
		\normalsize
		\subsection{Symmetry Groups of Sight Graphs}
		
		The symmetry group of $R^-$ is generated by two permutations: $\sigma=(18)(2\ 10\ 5\ 9)(3746)$ and $\tau=(25)(36)(47)$. It can be easily checked that these two permutations are indeed automorphisms of $R^-$. Let us show that they generate the full automorphism group $Aut(R^-)$. We can define blocks using the degrees of vertices: degree six vertices $\{1,8\}$, degree five vertices with a single degree six neighbour $\{2,5,9,10\}$ and degree five vertices with two degree six neighbours $\{3,4,6,7\}$. Since the first block is the only one of cardinality two it does not move. We already have a permutation that transpose $1$ and $8$, thus to find the full automorphism group it is sufficient to find all the permutations in $Stab(1)=Stab(8)$. Under the actions of this stabilizer, $\{2,5\}$ and $\{9,10\}$ are blocks that are fixed, since the former neighbours with $1$ and not $8$ and that the latter neighbours with $8$ and not $1$. Using the same argument again, since $\tau\in Stab(1)=Stab(8)$ transposes $2$ and $5$, it is sufficient to fully characterize the stabilizer of $2$ in the stabilizer of $1$, i.e. all the permutations that fix $1,2,5,8$. By the neighbourhood relations of the other six vertices, one sees that the only such permutation is $(37)(46)(9\ 10)$ (the mirror symmetry of left and right sides) but this is $\sigma\tau\sigma^{-1}$. Therefore, $\sigma$ and $\tau$ generate the whole group $Aut(R^-)$.

		The symmetry group of $R^L$ is same as that of $R$, namely $$G=\langle(25)(36)(47),(26)(35)(8\ 10),(24)(57)(89)\rangle.$$ It is clear that these three permutations are automorphisms of $R^L$. Let us show that they generate the whole $Aut(R^L)$. The vertex $1$ is unique since it is the Luthor vertex. This gives us the following block system: $\{1\}$, the other degree six vertices $\{8,9,10\}$, the degree five vertices $\{2,3,4,5,6,7\}$. Since the cardinalities of the blocks do not match they cannot move to each other. We can rotate and take mirror symmetries of Figure~\ref{graph-R-} to get any permutation we want on $8,9,10$. Therefore, to find all the elements of $Aut(G)$, it is sufficient to find all permutations that fix $8,9,10$. By neighbourhood relations it is clear that the only non-trivial such permutation is the first generator in the description of $G$ above. Therefore, $G$ is indeed $Aut(R^L)$.
	}
	\subsection{A Certificate Of \sfrac{1847}{276} Bound}
	The following are the 1920 inequalities with their respective non-zero coefficients derived from the rational solution of the shortened linear program, which has optimal value $\frac{1847}{276}$.
	\tiny

\end{document}